\documentclass[preprint,12pt]{elsarticle}

\usepackage{amssymb}
\usepackage{amsmath}
\usepackage{amsthm}
\usepackage{lineno}
\usepackage{hyperref}

\newdefinition{definition}{Definition}
\newdefinition{example}{Example}
\newdefinition{observation}{Observation}
\newtheorem{theorem}{Theorem}

\newtheorem{lemma}{Lemma}
\newtheorem{corollary}{Corollary}

\usepackage{graphicx}
\usepackage{tikz}
\usetikzlibrary{fit,patterns,decorations.pathmorphing,decorations.pathreplacing,calc}
\usepackage{xcolor}
\usepackage{amsfonts}
\usepackage{mathpartir}
\usepackage{multirow}
\usepackage{adjustbox}
\usepackage{color}

\newcommand{\sreset}{\mathcal{S}_r}

\newcommand{\dd}{\mathinner{.\,.}}

\renewcommand{\a}{\texttt{a}}
\renewcommand{\b}{\texttt{b}}
\renewcommand{\c}{\texttt{c}}
\newcommand{\str}[1]{\texttt{#1}}
\newcommand{\dol}{\text{\tt \$}}
\newcommand{\htag}{\text{\tt \#}}
\newcommand{\db}{\mathtt{dB}}
\newcommand{\sre}{\mathtt{sre}}
\newcommand{\infw}{\textbf{w}}
\newcommand{\bwt}{\mathtt{BWT}}

\newcommand{\rot}{\mathcal{R}}
\newcommand{\fact}{\mathcal{F}}

\newcommand{\vir}[1]{``#1''}

\newcommand{\grl}{g_{rl}}
\newcommand{\zno}{z_{no}}
\newcommand{\ze}{z_{e}}
\newcommand{\zend}{z_{end}}

\journal{Theoretical Computer Science}

\begin{document}

\begin{frontmatter}

\title{Smallest Suffixient Sets: \texorpdfstring{\\}{} Effectiveness, Resilience, and Calculation}

\author[1]{Hiroto Fujimaru\corref{cor1}} 
\ead{fujimaru.hiroto.134@s.kyushu-u.ac.jp}
\author[2,3]{Gonzalo Navarro\corref{cor1}} 
\ead{gnavarro@dcc.uchile.cl}
\author[4]{Giuseppe Romana\corref{cor1}} 
\ead{giuseppe.romana01@unipa.it}
\author[3,5]{Cristian Urbina\corref{cor1}} 
\ead{c.urbina-gallegos@uw.edu.pl}

\affiliation[1]{organization={Department of Information Science and Technology, Kyushu University},
            city={Fukuoka},
            country={Japan}}
\affiliation[2]{organization={Department of Computer Science, University of Chile},
            city={Santiago},
            country={Chile}}
\affiliation[3]{organization={Center for Biotechnology and Bioengineering (CeBiB)},
            city={Santiago},
            country={Chile}}
\affiliation[4]{organization={Department of Mathematics and Computer Science, University of Palermo},
            city={Palermo},
            country={Italy}}
\affiliation[5]{organization={Faculty of Mathematics, Informatics and Mechanics, University of Warsaw},
            city={Warsaw},
            country={Poland}}

\cortext[cor1]{Corresponding author}

\begin{abstract}
A suffixient set is a novel combinatorial object that captures the essential information of repetitive strings in a way that, provided with a random access mechanism, supports various forms of pattern matching. In this paper, we study the size $\chi$ of the smallest suffixient set as a repetitiveness measure.

First, we study its sensitivity to various string operations. We show that $\chi$ cannot increase by more than 2 after appending or prepending a character to the string. As a consequence, we are able to give simple linear-time online algorithms to compute smallest suffixient sets. We also show that, although reversing the string can increase $\chi$ by an arbitrary $O(n)$ value, it always holds $\chi(T)/\chi(T^R)\le 2$. We also prove lower and upper bounds for the additive or multiplicative increase of $\chi$ after applying arbitrary edit operations, or rotating the text. In particular, we show that the additive increase can be as large as $\Omega(\sqrt{n})$ for all those operations.

Secondly, we place $\chi$ among known repetitiveness measures. In particular, we show $\chi \le 2r$ (where $r$ is the number of runs in the Burrows-Wheeler Transform of the string), that there are string families where $\chi=o(v)$ (where $v$ is the size of the smallext lexicographic parse of the string), and that $\chi$ is uncomparable to almost all reachable measures based on copy-paste mechanisms. In passing, we give precise bounds for $\chi$ for some relevant string families, for example $\chi \le \sigma+2$ on episturmian words over alphabets of size $\sigma$ (e.g., $\chi \le 4$ on Fibonacci strings, for which we precisely characterize the only two smallest suffixient sets).
\end{abstract}

\end{frontmatter}


\section{Introduction}

The study of repetitive string collections has recently attracted considerable interest from the stringology community, triggered by practical challenges such as representing huge collections of similar strings in a way that they can be searched and mined directly in highly compressed form~\cite{Navacmcs20.3,Navacmcs20.2}. An example is the {\em European '1+ Million Genomes' Initiative}\footnote{{\tt https://digital-strategy.ec.europa.eu/en/policies/1-million-genomes}}, which aims at sequencing over a million human genomes: while this data requires around 750TB of storage in raw form (using 2 bits per base), the high similarity between human genomes would allow storing it in querieable form using two orders of magnitude less space.

An important aspect of this research is to understand how to measure repetitiveness, especially when those measures reflect the size of compressed representations that offer different access and search functionalities on the collection. Various repetitiveness measures have been proposed, from abstract lower bounds to those related to specific text compressors and indices; a relatively up-to-date survey is maintained~\cite{NavSurvey}. Understanding how these measures relate to each other sheds light on what search functionality is obtained at what space cost.

A relevant measure recently proposed is the size $\chi$ of the smallest {\em suffixient set} of the text collection~\cite{DGLMP23}, whose precise definition will be given later. Within $O(\chi)$ size, plus a random-access mechanism on the string, it is possible to support some text search 
functionalities, such as finding one occurrence of a pattern, or finding its maximal exact matches (MEMs), which is of central use on various bioinformatic applications~\cite{suffixientarrays}.

While there has been some work already on how to build minimal suffixient sets and how to index and search a string within their size, less is known about that size, $\chi$, as a measure of repetitiveness. It is only known~\cite{DGLMP23} that $\gamma \le \chi \le 2\overline{r}$ on every string family, where $\gamma$ is the size of the smallest {\em string attractor} of the collection (a measure that lower bounds most repetitiveness measures)~\cite{KP18} and $\overline{r}$ is the number of equal-letter runs of the Burrows-Wheeler Transform (BWT)~\cite{BW94} of the reversed string.
Very recently, it has been shown that $s = O(\chi)$, where $s$ is the size of the smallest Substring Equation System (SES) that uniquely describes the string \cite{SB26}.

In this paper we better characterize $\chi$ as a repetitiveness measure. First, we study how it behaves when the string (of length $n$) undergoes updates, showing in particular that it grows by $O(1)$ when appending or prepending symbols, but that it can grow additively by $\Omega(\sqrt{n})$ upon arbitrary edit operations or rotations, and by $\Omega(n)$ when reversing the string. However, for this last operation, we show that $\chi$ cannot grow by a factor larger than $2$. Second, we show that $\chi \le 2r$ on every string family, where $r$ is the number of equal-letter runs of the BWT of the string. We also show that there are string families where $\chi=o(v)$, where $v$ is the size of the smallest lexicographic parse~\cite{NOP21} (an alternative to the size of the Lempel-Ziv parse~\cite{LZ76}, which behaves similarly). In particular, this holds on the Fibonacci strings, where we fully characterize the only 2 smallest suffixient sets of size 4, and further prove that $\chi \le \sigma+2$ on all substrings of episturmian words over an alphabet of size $\sigma$. Since $v=O(r)$ on all string families, this settles $\chi$ as a strictly smaller measure than $r$, which is a more natural characterization than in terms of the reverse string. We also show that $\chi$ is incomparable with most ``copy-paste'' based measures~\cite{Navacmcs20.3}, as there are families where it is strictly smaller and others where it is strictly larger than any of those measures. 

This result relates to the important question of whether a measure $\mu$ is {\em reachable} (i.e., one can represent the string within $O(\mu)$ space), {\em accessible} (i.e., one can access any string position from an $O(\mu)$-size representation, in $o(n)$ time), or {\em searchable} (i.e., one can search for patterns in $o(n)$ time within space $O(\mu)$). Measure $r$ is, curiously, the only one to date being reachable and searchable, but not known to be accessible.
Now $\chi$ emerges as a measure smaller than $r$, which can search if provided with a mechanism to efficiently access substrings ($r$ does not need access to support searches). 
It has been recently shown that $\chi$ is reachable, since $s=O(\chi)$ is reachable \cite{SB26}, but its accessibility status remains unknown.

Our final contribution are new, extremely simple, {\em online} algorithms to compute smallest suffixient sets (and thus $\chi$), scanning the text left to right or right to left. While there already exist efficient algorithms to do this \cite{cop:spire2024}, our new algorithm can, at any point, exhibit a smallest suffixient set for the string it has just consumed. Just as online suffix tree constructions \cite{Wei73,Ukk95}, our algorithm uses $O(n)$ space and worst-case time in the transdichotomous RAM model over polynomial-size integer alphabets. The best previous algorithm obtains the same result \cite{cop:spire2024}, but it is not online and starts from the suffix array and other components of the suffix tree, whereas ours starts from the text and builds the suffix tree at the same time. All linear-time suffix tree construction algorithms run under the same model of computation.

A preliminary version of this paper appeared in the proceedings of the conference SPIRE 2025 \cite{NRU_SPIRE_2025}. In this extended version, we improved several of our previous results, establish new ones, and present the new online construction algorithms. We also show extended proofs, new figures, and more.

\section{Preliminaries}\label{sec:preliminaries}

An \emph{ordered alphabet} $\Sigma = \{a_1,\dots,a_\sigma\}$ is a finite set of symbols equipped with a total order $<$ such that $a_1 < a_2 < \dots < a_\sigma$.  When $\sigma = 2$, we assume $\Sigma = \{\a,\b\}$ with $\a < \b$.   
The special symbol $\dol$, if it appears, is always assumed to be the smallest of the alphabet. 

A \emph{string} $w[1\dd n]$ (or simply $w$ if it is clear from the context) of \emph{length} $|w| = n$ over the alphabet $\Sigma$ is a sequence $w[1]w[2]\cdots w[n]$ of symbols where $w[i] \in\Sigma$ for all $i \in [1,n]$.  The \emph{empty string} of length $0$ is denoted by $\epsilon$.  We denote by $\Sigma^*$ the set of all strings over $\Sigma$. Additionally, we let $\Sigma^+ = \Sigma^* \setminus \{\epsilon\}$ and $\Sigma^k = \{w\in\Sigma^* \mid |w| = k\}$. 
We denote by $w[i\dd j]$ the substring $w[i]w[i+1]\cdots w[j]$. If $x=x[1\dd n]$ and $y=y[1\dd m]$ are strings, we define the \emph{concatenation operation} applied on $x$ and $y$, as the string obtained by juxtaposing these two strings, that is, $x \cdot y = x[1] x[2]\cdots x[n]y[1] \cdots y[m] = xy$.
A string $x$ is a \emph{substring} of $w$ if $w = yxz$ for some $y,z \in \Sigma^*$. A string $x$ is a \emph{prefix} of $w$ if $w = x y$ for some $y \in \Sigma^*$. Analogously, $x$ is a \emph{suffix} of $w$ if $w = y x$ for some $y \in \Sigma^*$. We say that substrings, prefixes, and suffixes are \emph{non-trivial} if they are different from $w$ and $\epsilon$. The set of substrings of $w$ is denoted by $\fact_w$. We also let $\fact_w(k) = \fact_w\cap\Sigma^k$.
The \emph{reverse} of a finite string $w$ is the string $w^R = w[n]\cdot w[n-1]\cdots w[1]$.
We denote by $\rot(w)$ the multiset of rotations of $w[1\dd n]$, that is, $\rot(w) = \{w[i+1\dd n]w[1\dd i] \mid i \in [1\dd n]\}$. The \emph{Burrows-Wheeler transform} (BWT) of a string $w$, denoted $\bwt(w)$, is the transformation of $w$ obtained by collecting the last symbol of all rotations in $\rot(w)$ in lexicographic order. The \emph{BWT matrix} $B(w)$ of $w$ is the $(n \times n)$-matrix where the $i$-th row is the $i$-th rotation of $w$ in lexicographic order. 

A \emph{right-infinite string} $\infw$ ---we use \textbf{boldface} to emphasize its infinite length--- over $\Sigma$ is any infinite sequence $\mathbb{Z}^+ \rightarrow \Sigma$. The set of all infinite strings over $\Sigma$ is denoted $\Sigma^\omega$. A substring of $\infw$ is the finite string $\infw[i\dd j]$ for any $1\le i \le j$. A prefix of $\infw$ is a finite substring of the form $\infw[1\dd n]$ for some $n \ge 0$.  
The \emph{substring complexity function} $P_\infw(k): \mathbb{Z}^+\cup\{0\}\rightarrow \mathbb{Z}^+$ counts the number of distinct substrings of length $k$ in $\infw$, for any $k \in \mathbb{Z}^+\cup\{0\}$, that is,  $P_\infw(k) = |\fact_\infw(k)|$. For a finite string $w[1\dd n]$, the domain of $P_w$ is restricted to $[0\dd n]$.

\subsection{Measures of repetitiveness}
\begin{figure}[t]
    \centering
    \includegraphics[width=\textwidth]{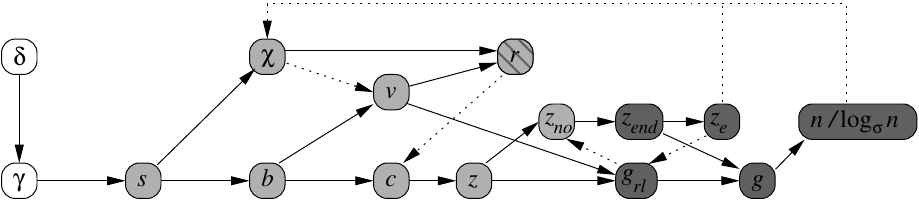}
    \caption{Relations between relevant repetitiveness measures and how our results place $\chi$ among them. An arrow $\mu_1 \to \mu_2$ means that $\mu_1=O(\mu_2)$ for all strings and, save for $s \to b$, $c \to z$, $z_{no} \to z_{end}$, and $z_{end} \to z_{e}$, a string family where $\mu_1=o(\mu_2)$ is known. The dotted arrows mark only this last condition, so they are not transitive. Measures in light gray nodes are known to be reachable; those in dark gray are accessible and searchable; and $r$ is hatched because it is searchable but not known to be accessible.}
    \label{fig:measures}
\end{figure}

In this work we will relate $\chi$, in asymptotic terms, with several well-established measures of repetitiveness~\cite{Navacmcs20.3,NavSurvey}:
$\delta = \max_{k \in [0\dd n]}(\fact_w(k)/k)$ (a measure of substring complexity), $\gamma$ (the smallest string attractor),
$s$ (the smallest Substring Equation System),
$b$ (the size of the smallest bidirectional macro scheme),
$z$ (the size of the Lempel-Ziv parse),
$\zno$ (the same without allowing phrases to overlap their sources),
$\ze$ (the size of the greedy LZ-End parse),
$\zend$ (the size of the minimal LZ-End parse),
$v$ (the size of the smallest lexicographic parse),
$r$ (the number of equal-letter runs in the BWT of the string),
$g$ (the size of the smallest context-free grammar generating only the string),
$\grl$ (the same allowing run-length rules), and
$c$ (the size of the smallest collage system generating only the string). Except for $\delta$, $\gamma$, and $r$, these measures are said to be {\em copy-paste} because they refer to a way of cutting the sequence into chunks that can be copied from elsewhere in the same sequence. Indeed, $\delta$ and $\gamma$ are lower-bound measures, the former known to be unreachable and the latter not known to date to be reachable; all the others are. The smallest measures known to be accessible (and searchable) are $\zend$ and $\grl$, and $r$ is searchable but not known to be accessible.

The known relations between those measures are summarized in Fig.~\ref{fig:measures}, where we have added the results we obtain in this paper with respect to $\chi$.

\subsection{Edit operations and sensitivity functions}

The so-called \emph{edit operations} are \emph{insertion}, \emph{substitution} and \emph{deletion} of a single character on a string. We denote $\mathtt{ins}(w)$, $\mathtt{sub}(w)$, and $\mathtt{del}(w)$ the sets of strings that can be obtained by applying an insertion, a substitution, and a deletion to $w$ respectively.  
In addition, we let $\mathtt{prepend}(w)$ and $\mathtt{append}(w)$ be $\mathtt{ins}(w)$ restricted to the insertion being made at the beginning and the end of the string, respectively, $\mathtt{rot}(w) = \mathcal{R}(w)$, and $\mathtt{rev}(w) = w^R$.

A repetitiveness measure $\mu$ is \emph{monotone} or \emph{non-decreasing} to the insertion of a single character if $\mu(w') - \mu(w) \ge 0$ for any $w$ and $w' \in  \mathtt{ins}(w)$. More generally, the \emph{additive sensitivity} and \emph{multiplicative sensitivity} functions of a repetitiveness measure $\mu$ to the insertion of a single character are the maximum possible values of $\mu(w') - \mu(w)$ and $\mu(w') / \mu(w)$, respectively. 
Formally, the additive sensitivity of a measure of repetitiveness $\mu$ to a string operation $\rho$ is defined as a function $AS_{\mu,\rho}:\mathbb{Z}^+\rightarrow\mathbb{R}$, where $AS_{\mu,\rho}(n)=\max_{w\in\Sigma^n} \max_{w'\in\rho(w)} \mu(w')-\mu(w)$, that is, the maximum achievable difference among all the strings. Similarly, the multiplicative sensitivity is defined as 
$MS_{\mu,\rho}(n)=\max_{w\in\Sigma^n} \max_{w'\in\rho(w)} \mu(w')/\mu(w)$.

We define the concept of monotonicity and sensitivity functions for the remaining string operations analogously.

\section{Suffixient Sets and the Measure \texorpdfstring{$\chi$}{chi}}

In this section we define the central combinatorial objects and measures we analyse on this work. Note that some of our definitions are slightly different from their original formulation~\cite{suffixientarrays,cop:spire2024}, because we do not always assume that all strings are $\dol$-terminated.

\begin{definition}[Right-maximal Substrings and Right-extensions~\cite{suffixientarrays,cop:spire2024}]
Let $w \in \Sigma^*$. A substring $x\in\Sigma^*$ of $w$ is \emph{right-maximal} if there exist at least two distinct symbols $a,b \in \Sigma$ such that both $xa$ and $xb$ are substrings of $w$. For any right-maximal substring $x$ of $w$, the substrings $xa$ with $a \in \Sigma$ are called \emph{right-extensions}. We denote the set of right-extensions in $w$ by $E_r(w) = \{xa\mid\exists b:b \neq a, xa \in \fact_w, xb \in \fact_w\}$.
\end{definition}

We distinguish a special class of right-extensions that are not suffixes of any other right-extension.

\begin{definition}[Supermaximal Extensions~\cite{suffixientarrays,cop:spire2024}]\label{def:supermaximal-extensions}
The set of \emph{supermaximal extensions} of $w$ is $\mathcal{S}_r(w) = \{x \in E_r(w)\mid\forall y\in E_r(w), y = zx \Rightarrow z = \varepsilon\}$.
Moreover, we let $\sre(w) = |\mathcal{S}_r(w)|$.
\end{definition}

We now define suffixient sets for strings not necessarily $\dol$-terminated; we introduce later the special terminator $\dol$.

\begin{definition}[Suffixient Set~\cite{suffixientarrays,cop:spire2024}]\label{def:suffixient}
Let $w[1\dd n] \in \Sigma^*$. A set $S \subseteq [1\dd n]$ is a \emph{suffixient set} for $w$ if for every right-extension $x \in E_r(w)$ there exists $j\in S$ such that  $x$ is a suffix of $w[1\dd j]$.
\end{definition}

Intuitively, a suffixient set is a collection of positions of $[1\dd |w|]$ capturing all the right-extensions appearing in $w$. The smallest suffixient sets, which are suffixient sets of minimum size, have also been characterized in terms of supermaximal right-extensions. The next definition simplifies the original one~\cite{suffixientarrays,cop:spire2024}.

\begin{definition}[Smallest Suffixient Set]\label{def:smallest-suffixient}
Let $w[1\dd n] \in \Sigma^*$. A suffixient set $S \subseteq [1\dd n]$ is a \emph{smallest suffixient set} for $w$ if there is a bijection $pos : \mathcal{S}_r(w) \to S$ such that every $x \in \mathcal{S}_r(w)$ is a suffix of $w[1\dd pos(x)]$.
\end{definition}

In its original formulation, the measure $\chi$ is defined over $\dol$-terminated strings. Here, we define $\chi(w)$ with the $\dol$ being implicit, not being part of $w$.

\begin{definition}[Measure $\chi$~\cite{suffixientarrays,cop:spire2024}]
Let $w \in \Sigma^*$ and assume $\dol \not \in \fact_w$. Then, $\chi(w) = |\mathcal{S}|$,  where $\mathcal{S}$ is a smallest suffixient set for $w\dol$.
\label{def:measure_chi}
\end{definition}

One can see from the above definitions that $\chi$ is well-defined because $\chi(w) = \sre(w\dol)$. We will use this relation to prove results on $\chi$ via $\sre$.

\section{Additive Sensitivity of \texorpdfstring{$\chi$}{chi} to String Operations}\label{sec:sensitivity}

The sensitivity to string operations has been studied for many repetitiveness measures~\cite{AFI23,FRSU23,FRSU25,GILPST21,GILRSU25,attractors_combinatorics,NOU25,NU25}. A robust repetitiveness measure should not change much upon small changes in the sequence. For instance, $b$, $z$, and $g$ can increase only by a multiplicative constant after an edit operation~\cite{AFI23}, and they can increase only by a constant additive factor when prepending or appending a character. On the other hand, $r$ can increase by a $\Theta(\log n)$ factor when appending a character~\cite[Prop. 37]{GILRSU25}. Other results have been obtained concerning more complex string operations, like reversing a string~\cite{GILPST21}, or applying a string morphism~\cite{FRSU23,FRSU25}. 

In this section we study how $\sre$ and $\chi$ behave in this respect, focusing on additive sensitivity; the next section deals with multiplicative sensitivity.
Overall, we obtain the following results on the additive sensitivity of $\chi$; the proof is given at the end of the section.

\begin{theorem}\label{co:AS_operations}The following bounds on the additive sensitivity of measure $\chi$ to string operations hold:
\begin{enumerate}
    \item $AS_{\chi,\rho}(n) = \Theta(1)$ for $\rho \in \{ \mathtt{append}, \mathtt{prepend} \}$;
    \item $AS_{\chi,\rho}(n)=\Omega(\sqrt n)$ for $\rho \in \{\mathtt{ins},\mathtt{del},\mathtt{sub}, \mathtt{rot}\}$;
    \item $AS_{\chi,\mathtt{rev}}(n)=\Theta(n)$.
\end{enumerate}
\end{theorem}

\subsection{Appending and prepending symbols}

We first show that $\sre(w)$ and $\chi(w)$ increase or decrease by only additive constants when adding or removing symbols at the extremes of $w$. We start by proving the following useful lemma.

\begin{lemma}\label{le:re_containment}
If $E_r(w_1) \subseteq E_r(w_2)$, then $\sre(w_1) \le \sre(w_2)$.
\end{lemma}

\begin{proof}Let $x,y \in \sreset(w_1)$ with $x \neq y$. Because $x \in E_r(w_2)$, there exists $z \in \sreset(w_2)$ with $x$ a suffix of $z$. Because $y$ is not a suffix of $x$ and vice versa, $y$ cannot be a suffix of $z$. Therefore, the map $x \mapsto z$ with $x\in \sreset(w_1)$, $z \in \sreset(w_2)$, and $z = z'x$ for some $z' \in \Sigma^*$ is injective and then $\sre(w_1) \le \sre(w_2)$. 
\end{proof}

We now prove that $\sre(w)$ grows only by $O(1)$ when prepending or appending characters.

\begin{lemma}\label{lem:sre_append_bounded}
Let $w\in \Sigma^*$, and $c \in \Sigma$. It holds $\sre(w) \le  \sre(wc) \le \sre(w)+2$.
\end{lemma}

\begin{proof}
The lower bound follows from Lemma \ref{le:re_containment}.
For the upper bound, we analyse the new right-extensions that may arise due to appending $c$ to $w$. For any fixed suffix $xc$ of $wc$:
\begin{enumerate}
\item if $xc$ appears in $w$, or if $xa$ does not appear in $w$ for any $a \neq c$, or both, then $xc$ induces no new right-extensions in $wc$;
\item if for some $a \neq b$, $xa$ and $xb$ were both substrings of $w$, and $xc$ was not, then $xc$ is a new right-extension of $wc$;
\item if $x$ is always followed by $a\neq c$ in $w$ (hence, $xa$ is not a right-extension of $w$), then both $xa$ and $xc$ are new right-extensions of $wc$.
\end{enumerate}

Cases 1 and 2 induce at most one new supermaximal right-extension in total for all possible $xc$, namely the longest right-extension in $wc$ that is a suffix of $wc$.
For Case 3, consider all the increasing-length suffixes \sloppy $x_1c, x_2c, \dots, x_tc$ of $wc$ that became right-extensions together with $x_1a, x_2a, \dots, x_ta$. Since the latter form a chain of suffixes of $x_ta$, we only have one possible new supermaximal right-extension ending with $a$, namely $x_ta$. 
Observe that the chain of suffixes $x_1a, x_2a, \dots, x_ta$ is unique: if the suffix $x$ is always followed by $a$, any suffix $y$ of $x$ is either right-maximal in $w$ (and $y$ falls within Case 2), or it is always followed by an $a$ (because $x$ is always followed by an $a$), i.e. $y=x_i$ for some $i\in[1\dd t]$.  
\end{proof}

\begin{lemma} \label{lem:sre_prepend_bounded}
Let $w\in \Sigma^*$ and $c \in \Sigma$. It holds $\sre(w) \le  \sre(cw) \le \sre(w) + 2$.
\end{lemma}

\begin{proof}The lower bound follows from Lemma \ref{le:re_containment}. For the upper bound, let $cxa$ be the shortest prefix of $cw$ that is not a right-extension of $w$, but is a right-extension of $cw$ (if it exists). This implies that there exists $b\neq a$ such that $cxb \in \fact_w$ and $cxa \not \in \fact_w$ (otherwise, $cxa$ would be a right-extension of $w)$, so no prefix of $cw$ of length $|cxa|$ or more is right-maximal, and thus no prefix longer than $cxa$ can be a right-extension. By the minimality of $cxa$, all prefixes of $cw$ shorter than $cxa$ either are not right-extensions of $cw$, or are right-extensions of both $w$ and $cw$. Therefore, $cxa$ together  with some $cxb$ appearing in $w$, are the only possible new right-extensions in $cw$ with respect to $w$.
\end{proof}

Lemmas~\ref{lem:sre_append_bounded} and \ref{lem:sre_prepend_bounded} yield remarkably simple algorithms to compute smallest suffixient sets, which we detail in Section~\ref{sec:ukkonen}. 

\medskip

By letting $c = \dol \not \in \fact_w$ in Lemma~\ref{lem:sre_append_bounded},  we relate $\chi$ to $\sre$ (note that $\chi$ is always at least $\sre + 1$ because of the new supermaximal extension ending with $\dol$). This makes clear the relation between Combinatorics on words~\cite{LothaireAlg} with suffixient sets, via the common notion of \emph{right-special factors} (which we call here right-maximal substrings).

\begin{corollary}\label{cor:sre_chi}
 Let $w \in \Sigma^*$. 
 It holds $\sre(w)+1 \le \chi(w) \le \sre(w)+2$.
\end{corollary}

Note that, while the value $\sre(w)$ is non-decreasing after appending a character, this is not the case for the measure $\chi$.

\begin{lemma}\label{le:chi_non_monotone}
The measure $\chi$ is not monotone to appending a character.
\end{lemma}

\begin{proof}Let $w = \str{abaab}$.
It holds $\sreset(w\dol)  = \{\str{aa}, \str{ab}, \str{ab\dol}, \str{aba}\}$ and $\sreset(w\a\dol) = \sreset(\str{abaaba\dol}) = \{\str{ab}, \str{aba\dol}, \str{abaa}\}$. 
Hence, $\chi(w) = 4$ and $\chi(w\a) = 3$. 
\end{proof}

\subsection{Sensitivity to edit operations and rotations}

We first prove some asymptotic relations between the additive sensitivities of edit operations and rotations.

\begin{lemma} \label{lem:edit-vs-rot-additive}
$AS_{\chi,\rho}(n)=\Theta(AS_{\chi,\mathtt{rot}}(n))$ for $\rho \in \{\mathtt{del},\mathtt{sub}\}$, and $AS_{\chi,\mathtt{ins}}(n)=O(AS_{\chi,\mathtt{rot}}(n))$.
\end{lemma}
\begin{proof}
We start proving $AS_{\chi,\rho}(n)=O(AS_{\chi,\mathtt{rot}}(n))$, with $\rho=\mathtt{ins}$.
Let $|xy|=n$ and let us insert $c\in\Sigma$ between $x$ and $y$. Let us also shorten $R(n) = AS_{\chi,\mathtt{rot}}(n)$. Then it holds that
\begin{eqnarray*}
\sre(xcy) &\le& \sre(yxc) + R(n+1) \le \sre(yx) + 2 + R(n+1)\\
&\le& \sre(xy) + R(n)+2+R(n+1),
\end{eqnarray*}
where the second inequality stems from Lemma~\ref{lem:sre_append_bounded}. Since $1 \le R(n) = O(n)$ because $\sre = O(n)$, it follows that $\sre(xcy) = \sre(xy) + O(AS_{\chi,\mathtt{rot}}(n))$. 

For $\rho=\mathtt{del}$, we instead delete $d$ from $xdy$, so we have
$$\sre(xy) \le \sre(yx) + R(n) \le \sre(yxd) + R(n)
\le \sre(xdy) + R(n+1)+R(n),
$$
where again the second inequality stems from Lemma~\ref{lem:sre_append_bounded}. By combining both equations above, we ge a bound for $\rho=\mathtt{sub}$, with $n=|xcy|$:
$$
\sre(xcy) \le \sre(yx) +2+R(n) \le 
\sre(yxd)+4+R(n) \le \sre(xdy)+4+2R(n).
$$

In the other direction, for $AS_{\chi,\mathtt{rot}}(n)=O(AS_{\chi,\mathrm{del}}(n))$, let $\dol$ be a symbol not in $xy$ and $D(n)=AS_{\chi,\mathtt{del}}(n)$. Then
$$
\sre(yx) \le \sre(y\dol x) + D(n+1) \le \sre(xy\dol) + D(n+1) \le \sre(xy) + 2+D(n+1).
$$
The second inequality stems from $\dol$ being unique in $y\dol x$, so no right-maximal string can contain it. Thus, $E_r(y \dol x) = E_r(y\dol) \cup E_r(x) \cup  U\subseteq E_r(xy\dol)$, where $U=\{ua,ub\,|\,a \neq b \land ua\in \fact_x \land ub \in \fact_{y\dol}\}$. By Lemma~\ref{le:re_containment}, this implies $\sre(y\dol x) \le \sre(xy\dol)$. For substitutions, let $S(n)=AS_{\chi,\mathrm{sub}}(n)$, and consider strings $x$ and $yc$, with $c \in \Sigma$ and $n=|ycx|$:
\begin{eqnarray*}
\sre(ycx) &\le& \sre(y\dol x) + S(n) \le \sre(xy\dol) + S(n) \\ 
&\le& \sre(xy) + 2+S(n) \le \sre(xyc) + 2 + S(n).\qquad\qquad\qquad\qedhere
\end{eqnarray*}
\end{proof}

We now show that $\sre$ can grow by $\Omega(\sqrt{n})$ upon arbitrary edits and rotations.

\begin{lemma}
\label{lem:chi_of_w}
 Let $w_m = \a\b^{2m}\a\b^{2m+2}\prod_{k=1}^{m-1}u_k$, where $u_k = (\a\b^k\a\b^{2m-k})^2$ for some $m\geq3$. Then, $\sre(w_m) = 6m+4$ holds.
\end{lemma}

\begin{proof}
Let us consider the prefix $v = \a\b^{2m}\a\b^{2m+2}$ of $w$. 
 The substrings $\a\b^{2m}$ and $\b^{2m+1}$ appear twice in $v$, once followed by $\a$ and once followed by $\b$.
 For each $k\in[1,m-1]$: the substring $\b^k\a\b^{2m-k}$ appears three times, once in $v$, and twice in $u_k$ (note that one occurrence of $\b^{m-1}\a\b^{m+1}$ is a suffix of $w_m$); the substring $\b^{2m-k+1}\a\b^{k}$ appears once in $v$ and twice across $u_k u_{k+1}$;  the substring $\b^{2m-k}\a\b^{k}\a\b^{2m-k}\a\b^{k}$ appears twice across $u_{k-1} u_k u_{k+1}$. 
 
 By carefully analyzing the right-extensions of these factors, one can verify that the following substrings are supermaximal right-extensions of $w_m$:
 \begin{enumerate}
     \item  \label{bullet:additive1}$\a\b^{2m}\a$ and $\a\b^{2m+1}$,
     \item \label{bullet:additive2} $\b^k\a\b^{2m-k}\a$ and $\b^k\a\b^{2m-k+1}$ (for $k\in[1,m-1]$),
     \item \label{bullet:additive3}$\b^{2m-k+1}\a\b^{k}\a$ and $\b^{2m-k+1}\a\b^{k+1}$ (for $k\in[1,m-1]$),
     \item \label{bullet:additive4}$\b^{2m-k}\a\b^{k}\a\b^{2m-k}\a\b^{k}\a$ and $\b^{2m-k}\a\b^{k}\a\b^{2m-k}\a\b^{k+1}$ (for $k\in[1,m-2]$),
     \item \label{bullet:additive5}$\b^{2m+1}\a$ and $\b^{2m+2}$.
 \end{enumerate}

 We now prove that there are no other super-maximal right-extensions in $w$.
 Let us consider the other right-extensions (also see Figure~\ref{fig:chi_of_w}). For any other right-extensions $x$ ending in $\a$, we have that it is a suffix of one of the supermaximal right-extensions, since every occurrence of $\a$, except the first and last $\a$, serves as the ending position of a supermaximal right-extension listed above; the longest right-extensions ending in correspondence of the first and last $\a$ are $\a$ and $\b^{m+1}\a\b^{m-1}\a$ respectively, and both are suffixes of $\b^{m+2}\a\b^{m-1}\a$ (case~\ref{bullet:additive3}. with $k=m-1$).
 For any other right-extension ending in $\b$: by construction, any substring including $\a$ at least twice is not a right-extension, with the exception of $\b^{2m-k}\a\b^{k}\a\b^{2m-k}\a\b^{k+1}$ for $k\in[1,m-2]$ (included in case~\ref{bullet:additive4}.); 
 all right-extensions containing only one $\a$ have the form $\b^i\a\b^j$, for any $i\geq0$,$j\geq 1$ such that $i+j\leq 2m+1$ and $j\ne m$, and each is a suffix of one of the supermaximal right-extensions included in either cases~\ref{bullet:additive2}.\ or~\ref{bullet:additive3}.; finally, the right-extensions that do not contain $\a$'s consist of runs of $\b$'s, which are all suffixes of $\b^{2m+2}$ (included in~\ref{bullet:additive5}.).
Therefore, there are no other supermaximal extensions, and $\sre(w_m)=6m-4$.
\end{proof}

 \begin{figure}[t]
    \centering
    \includegraphics[width=\textwidth]{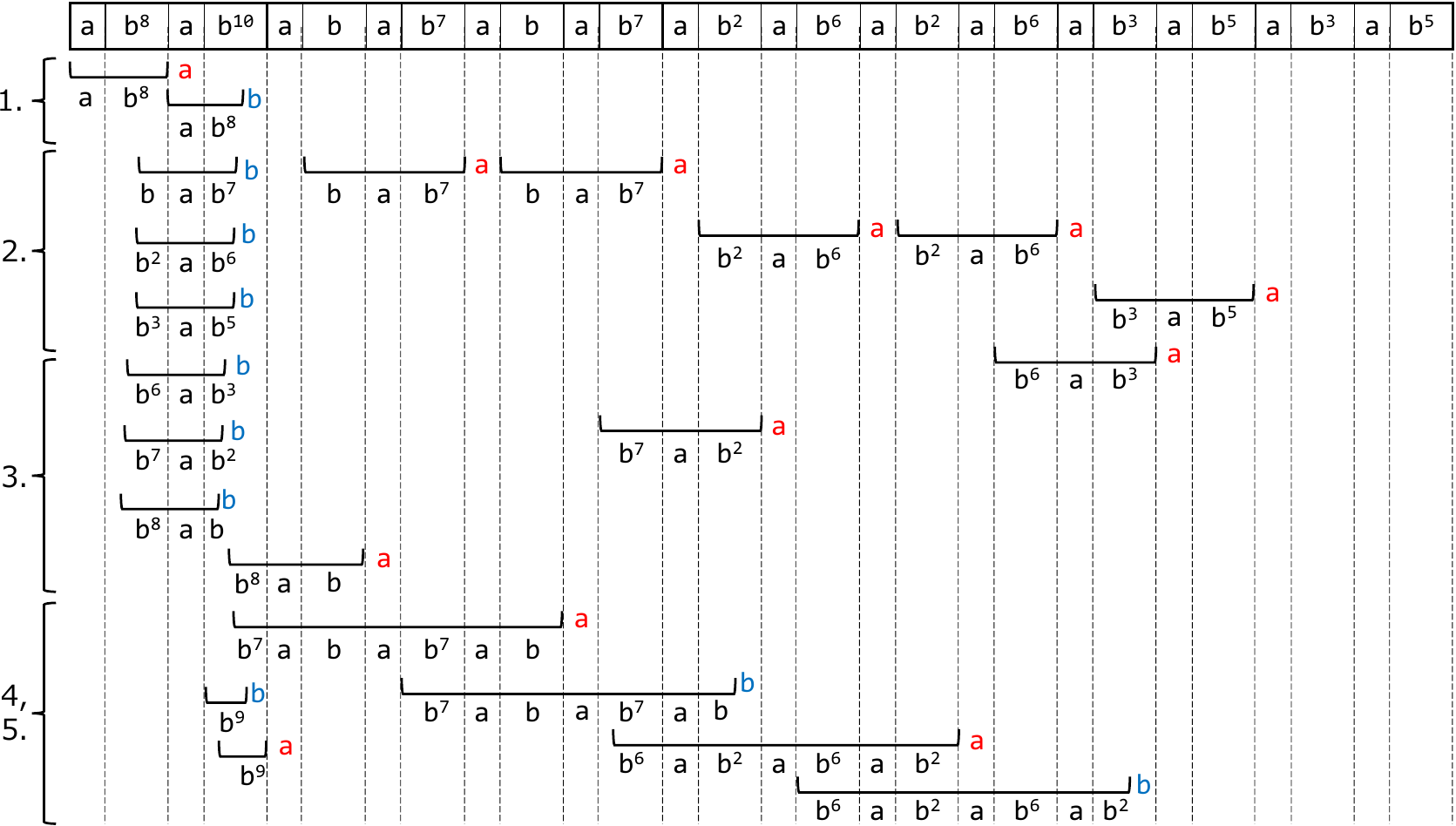}
    \caption{The supermaximal right-extension of the string $w_4$ shown in Example \ref{ex:chi_of_w}. Note that the ending positions of those ending within $\b^{10}$ are all different. }
    \label{fig:chi_of_w}
\end{figure}

 \begin{example}
 \label{ex:chi_of_w}
 Let $w_4 = \a\b^8\a\b^{10}\a\b\a\b^7\a\b\a\b^7\a\b^2\a\b^6\a\b^2\a\b^6\a\b^3\a\b^5\a\b^3\a\b^5$ (also see Figure \ref{fig:chi_of_w}).
 It can be verified that the supermaximal right-extensions of $w_4$ are:
    \begin{enumerate}
     \item $\a\b^8\a$ and $\a\b^9$,
     \item $\b\a\b^7\a$ and $\b\a\b^8$; $\b^2\a\b^6\a$ and $\b^2\a\b^7$; $\b^3\a\b^5\a$ and $\b^3\a\b^6$,
     \item $\b^8\a\b\a$ and $\b^8\a\b^2$; $\b^7\a\b^2\a$ and $\b^7\a\b^3$; $\b^6\a\b^3\a$ and $\b^6\a\b^4$,
     \item $\b^7\a\b\a\b^7\a\b\a$ and $\b^7\a\b\a\b^7\a\b^2$; $\b^6\a\b^2\a\b^6\a\b^2\a$ and $\b^6\a\b^2\a\b^6\a\b^3$,
     \item $\b^9\a$ and $\b^{10}$.
    \end{enumerate}
 One can see that $\sre(w_4) = 20$, as stated in Lemma \ref{lem:chi_of_w}.
 \end{example}

We next provide the lower-bound families for edit operations and rotations. When convenient, we describe the witness string via the reverse transformation; the labels (Ins), (Del), (Sub), and (Rot) refer to the operations whose sensitivity is being witnessed.

\begin{lemma} \label{lem:new_edit_additive_sensitivity}
 Let $w_m = \a\b^{2m}\a\b^{2m+2}\prod_{k=1}^{m-1}u_k$, where $u_k = (\a\b^k\a\b^{2m-k})^2$ for some $m\geq3$. It holds:
\begin{enumerate}
\item[(Ins)] If $w_m' = \a\b^{4m+2}\prod_{k=1}^{m-1}u_k$, then $\sre(w_m') = 4m - 2$;
\item[(Sub)] If $w_m' = \a\b^{4m+3}\prod_{k=1}^{m-1}u_k$, then $\sre(w_m') = 4m - 2$;
\item[(Del)] If $w_m' = \a\b^{2m}\a\a\b^{2m+2}\prod_{k=1}^{m-1}u_k$, then $\sre(w_m') = 4m + 2$;
\item[(Rot)] If $w_m' = \a\b^{2m+2}\prod_{k=1}^{m-1}u_k\cdot\a\b^{2m}$, then $\sre(w_m') = 4m - 2$.
\end{enumerate}
\end{lemma}

\begin{proof}
We prove the claims separately.

\noindent\textbf{Insertions.}
Let us consider the string $w_m'$ obtained by removing the second $\a$ from $w_m$. 
Then, the prefix $v=\a\b^{2m}\a\b^{2m+2}$ of $w_m$ is changed to $v'=\a\b^{4m+2}$. The substring $\b^{2m-k+1}\a\b^k$ occurs only once for each $k\in[1..m-1]$, and any subtring $\b^i\a\b^j$ of $w_m$ with $i+j >2m$ does not occur in $w'_m$. 
Instead, $\b^{4m+1}$ occurs twice in $v'$, and $\b^k\a\b^{2m-k-1}\a$ and $\b^k\a\b^{2m-k}$ become new supermaximal right-extensions for each $k\in [1,m-2]$, occurring twice in $u_{k+1}$ and $u_k$, respectively. Also, $\a\b^{2m-1}$ appears once in $v'$ and twice in $u_1$, and $\b^{m-1}\a\b^{m-1}$ can be found three times in $u_{m-1}$. Wrapping up and analyzing its right-extensions, one can verify that the supermaximal right-extensions of $w_m'$ and their ending positions are the following:
\begin{enumerate}
     \item $\a\b^{2m-1}\a$ (ending at $(2m+3)$ in $v'$), and $\a\b^{2m}$ (ending at $(2m+3)$ in $u_1$ and $1$ in $u_2$),
     \item $\b^k\a\b^{2m-k-1}\a$ (ending at $(2m+3)$ in $u_{k+1}$ and $1$ in $u_{k+2}$) and $\b^k\a\b^{2m-k}$ (ending at $(2m+2)$ and $(4m+4)$ in $u_k$) for $k\in[1,m-2]$,
     \item $\b^{2m-k}\a\b^{k}\a\b^{2m-k}\a\b^{k}\a$ (ending at $(2m+k+4)$ in $u_k$)and $\b^{2m-k}\a\b^{k}\a\b^{2m-k}\a\b^{k+1}$ (ending at $(k+2)$ in $u_{k+1}$) for $k\in[1,m-2]$,
     \item $\b^{4m+1}\a$ (ending at $(4m+4)$ in $v'$) and $\b^{4m+2}$ (ending at $(4m+3)$ in $v'$),
     \item $\b^{m-1}\a\b^{m-1}\a$ (ending at $(2m+1)$ in $u_{m-1}$) and $\b^{m-1}\a\b^m$ (ending at $(m+1)$ and $(3m+3)$ in $u_{m-1}$) .
 \end{enumerate}
Thus, we have $\sre(w_m') = 4m-2$.

\noindent\textbf{Substitutions.}
Assuming that the string $w_m'$ is obtained by replacing the second $\a$ in $w_m$ with $\b$, we obtain a string similar to the one we consider in case~\textbf{Insertions}. Considering its structure, we have that the supermaximal right-extensions of $w_m'$ are the following:

\begin{enumerate}
     \item $\a\b^{2m-1}\a$ and $\a\b^{2m}$,
     \item $\b^k\a\b^{2m-k-1}\a$ and $\b^k\a\b^{2m-k}$ (for $k\in[1,m-2]$),
     \item $\b^{2m-k}\a\b^{k}\a\b^{2m-k}\a\b^{k}\a$ and $\b^{2m-k}\a\b^{k}\a\b^{2m-k}\a\b^{k+1}$ (for $k\in[1,m-2]$),
     \item $\b^{4m+2}\a$ and $\b^{4m+3}$,
     \item $\b^{m-1}\a\b^{m-1}\a$ and $\b^{m-1}\a\b^m$.
 \end{enumerate}
Thus, we have $\sre(w_m') = 4m-2$.

\noindent\textbf{Deletions.}
Let us consider the string $w_m'$ obtained by adding $\a$ just following the second $\a$ of $w_m$. 
Same as above, $\b^{2m-k+1}\a\b^k$ occurs only once for each $k\in[1..m-1]$, and any other occurrences of $\b^i\a\b^j$ disappears when $i+j >2m$.
This is almost as the previous case, but we have to consider the right-extension of $\b^{2m}\a$ found twice, once in the prefix $v'=\a\b^{2m}\a\a\b^{2m+2}$ and once across $v'u_1$, whereas $\a\b^{2m}\a$ and $\a\b^{2m+1}$ still remain supermaximal right-extensions. We then have that the elements of $\sreset(w_m')$ are the following:
\begin{enumerate}
     \item $\a\b^{2m}\a$ and $\a\b^{2m+1}$,
     \item $\a\b^{2m-1}\a$ and $\a\b^{2m}$,
     \item $\b^k\a\b^{2m-k-1}\a$ and $\b^k\a\b^{2m-k}$ (for $k\in[1,m-2]$),
     \item $\b^{2m-k}\a\b^{k}\a\b^{2m-k}\a\b^{k}\a$ and $\b^{2m-k}\a\b^{k}\a\b^{2m-k}\a\b^{k+1}$ (for $k\in[1,m-2]$),
     \item $\b^{2m}\a\a$ and $\b^{2m}\a\b$,
     \item $\b^{2m+1}\a$ and $\b^{2m+2}$,
     \item $\b^{m-1}\a\b^{m-1}\a$ and $\b^{m-1}\a\b^m$.
 \end{enumerate}
Thus, we have $\sre(w_m') = 4m+2$.

\noindent\textbf{Rotations.}
Lastly, we consider the string $w_m'$ that is a rotation of $w_m$ beginning in the second $\a$. 
Note that $u_m$ is now followed by $\a\b^{2m}$.
Hence, with respect to $w_m$, the substring $\b^{m+1}\a\b^{m-1}\a\b^{m+1}\a\b^{m-1}$ occurs twice in $w'_m$. 
Thus, we have that the elements of $\sreset(w_m')$ are the following:
\begin{enumerate}
     \item $\b^{2m+1}\a$ and $\b^{2m+2}$,
     \item $\b^k\a\b^{2m-k-1}\a$ and $\b^k\a\b^{2m-k}$ (for $k\in[1,m-1]$),
     \item $\b^{2m-k}\a\b^{k}\a\b^{2m-k}\a\b^{k}\a$ and $\b^{2m-k}\a\b^{k}\a\b^{2m-k}\a\b^{k+1}$ (for $k\in[1,m-1]$).
 \end{enumerate}
 Thus, we have $\sre(w_m') = 4m-2$.
 \end{proof}

\begin{corollary} \label{lem:8}
There exists a string family where $\sre(w)-\sre(w') \in \Omega(\sqrt{n})$.
\end{corollary}
\begin{proof}
The strings $w_m$ in Lemmas~\ref{lem:chi_of_w} and \ref{lem:new_edit_additive_sensitivity} have length $\Theta(m^2)$. Since $\sre(w_m) = 6m+4$ and $\sre(w_m') \le 4m+2$ for the string $w_m'$ that results from applying any edit operation or rotation on $w_m$, the result follows.
\end{proof}

\subsection{Sensitivity to string reversal}

We now show that $\sre$ can grow by $\Omega(n)$ upon string reversals. We remind that the maximum additive growth is always $O(n)$ because $\chi \le n$.

\begin{lemma}\label{le:reverse_additive_sensitivity}There exists a string family where $\sre(w^R) -\sre(w) = n/5-1$.
\end{lemma}

\begin{proof}
Such a family is composed of the strings
$w_k= \prod_{i=1}^k \c \a_i\htag_i \a_i\dol_i$ on the alphabet $\Sigma = \bigcup_{i=1}^k\{\a_i,\htag_i,\dol_i\}\cup\{\c\}$. The size of $w_k$ is $  n = 5k$.

Observe that any substring of $w_k$ containing $\htag_i$ or $\dol_i$ is not right-maximal because those symbols occur once.  Moreover, every substring of length $2$ of $w_k$ is also unique and not right-maximal. Therefore, all the supermaximal extensions in $w_k$ and $w_k^R$ have length at most $2$. 

One can verify that the set of supermaximal extensions of $w_k$ is $\sreset(w_k) =\bigcup_{i=1}^k\{\a_i\htag_i, \a_i\dol_i, \c\a_i\} \cup\{\c\}$. Thus, $\sre(w_k)=3k+1$.

On the reversed string, one can verify that the set of supermaximal extensions of $w_k^R$ is $\sreset(w_k^R) =\bigcup_{i=1}^k\{\a_i\c, \a_i\htag_i, \a_i\} \cup \bigcup_{i=1}^{k-1}\{\c\dol_i\}\cup\{\dol_k\}$. Thus, $\sre(w_k^R)=4k$, and the claim follows.
\end{proof}

\begin{proof}[Proof (of Theorem~\ref{co:AS_operations})]
    We have obtained along the section those results in terms of $\sre$, which by Corollary \ref{cor:sre_chi} can be written in terms of $\chi$.
    Claim 1 follows by Lemmas~\ref{lem:sre_append_bounded} and~\ref{lem:sre_prepend_bounded}.
    Claim 2 follows by Lemma~\ref{lem:new_edit_additive_sensitivity}, where $n=|w_m|=\Theta(m^2)$ and $AS_{\chi,\rho}(n) = \Omega(m) = \Omega(\sqrt{n})$, for all $\rho \in \{\mathtt{ins},\mathtt{del},\mathtt{sub}, \mathtt{rot}\}$.
    The $\Omega(n)$ part of Claim 3 follows by Lemma~\ref{le:reverse_additive_sensitivity}, where $n=|w_k|=5k=\Theta(k)$ and $AS_{\chi,\mathtt{rev}}(n) = \Omega(k) = \Omega(n)$; the $O(n)$ part is a consequence of $\chi$ being $O(n)$.
\end{proof}

\section{Multiplicative Sensitivity of \texorpdfstring{$\chi$}{chi} to String Operations}

We now focus our attention to multiplicative sensitivity. We first show that $\chi$ has a multiplicative sensitivity to reversals that is a constant strictly larger than $1$. Then we show a relation between the multiplicative sensitivities of the other operations. We finally express both additive and multiplicative sensitivities in terms of the substring complexity $\delta$.

\subsection{Sensitivity to edits and rotations}

A simple consequence of previous results is that the multiplicative sensitivity of edits and rotations is at least $3/2 -  o(1)$.

\begin{corollary} \label{lem:12}
There exists a string family where $\lim_{n\rightarrow \infty }\sre(w)/\sre(w')=3/2$ upon edit operations or rotations.
\end{corollary}
\begin{proof}
The strings $w_m$ in Lemmas~\ref{lem:chi_of_w} and \ref{lem:new_edit_additive_sensitivity} satisfy $\sre(w_m) = 6m+4$, whereas  $\sre(w_m') \le 4m+2$ holds for the string $w_m'$ that results from applying any edit operation or rotation on $w_m$. The result follows.
\end{proof}

We conjecture that the multiplicative sensitivity of all those operations is indeed constant. The following results relate them all in this sense.

\begin{lemma}For any  $\mathtt{op}_1,\mathtt{op}_2\in \{\mathtt{ins},\mathtt{del},\mathtt{sub},\mathtt{rot}\}$, $\sre(\mathtt{op}_1(w))/\sre(w)=\Theta(1)$ if and only if $\sre(\mathtt{op}_2(w))/\sre(w)=\Theta(1)$.
\end{lemma}

\begin{proof}
Suppose $\sre(w')/\sre(w)=\Theta(1)$ for any $w=xy$ and $w'=yx$ a rotation of $w$. Then, writing $X \approx Y$ for $X = \Theta(Y)$, it holds 
$$\sre(xcy) \approx \sre(yxc) \approx \sre(yx) \approx \sre(xy)$$
and
$$\sre(xcy) \approx \sre(yxc) \approx \sre(yxd) \approx \sre(xdy),$$
which implies $\sre(w')/\sre(w)=\Theta(1)$ for all edit operations. This holds because appending symbols change $\sre$ by $O(1)$ only.

Now assume $\sre(w')/\sre(w)=\Theta(1)$ for any $w=xy$ and $w'=xcy$ (insertion) or vice versa (deletion). Then, using a $\dol$ as in Lemma~\ref{lem:edit-vs-rot-additive},
$$\sre(xy) \approx \sre(x\dol y) \le \sre(yx\dol) = O(\sre(yx)),$$
$$\sre(yx) \approx \sre(y\dol x) \le \sre(xy\dol) = O(\sre(xy)),$$
and thus $\sre(xy) \approx \sre(yx)$. Finally, assume $\sre(w')/\sre(w)=\Theta(1)$ for any $w=xcy$ and $w'=xdy$ (substitution). Then
$$\sre(xcyd) \approx \sre(x\dol yd) \le \sre(ydx\dol) = O(\sre(ydxc)),$$
$$\sre(ydxc) \approx \sre(y\dol xc) \le \sre(xcy\dol) = O(\sre(xcyd)),$$
thus $\sre(xcyd) \approx \sre(ydxc)$, which includes every rotation.
\end{proof}

\subsection{Sensitivity to string reversals}

The proof of Lemma~\ref{le:reverse_additive_sensitivity} immediately yields a lower bound of $4/3-o(1)$.

\begin{corollary}\label{cor:reverse_multiplicative_sensitivity}
There exists a string family where $\lim_{n\rightarrow \infty }\sre(w^R)/\sre(w)=4/3$.
\end{corollary}

\begin{proof}
For the family described in the proof of Lemma~\ref{le:reverse_additive_sensitivity}, it holds that 
$\sre(w_k)=3k+1$ and $\sre(w_k^R)=4k$.
\end{proof}

We now prove that $\chi$ at most doubles when we reverse the string, by studying how $\sre$ behaves.
For simplicity of analysis, we assume that $w$ contains a special symbol $\dol$ at the beginning and at the end (this works because it is not hard to see that $\chi(w) = \sre(w\$) = \sre(\dol w\dol)$ for every $w$). We do this to (i) ensure all right-maximal substrings of $w$ have at least two left-extensions, and hence, right-extensions have at least one left-extension, and (ii) provide the needed symmetry to prove Lemma~\ref{le:reverse_multiplicative} below. 

In the following, we denote the set of symbols preceding a substring $x$ (or \emph{left-extensions} of $x$) of $w$ by $L_x(w)$, and we use $l_x(w)$ to denote an arbitrary element of $L_x(w)$ (we drop $w$ if it is clear from the context). 
Please note that left-extensions are not analogous to right-extensions. A right-extension is a string $xa$ such that $xb$ also occurs for some $b \neq a$, while a left-extension of $x$ is just any symbol preceding $x$. A left-extension can be unique, whereas right-extensions come  (at least) in pairs.
We will rely on the following observation.

\begin{observation} \label{obs:Ldisjoint}
If $xa$ is a supermaximal right-extension, then $L_{xa} \cap L_{xb}= \emptyset$ for any $b \neq a$. 
\end{observation}
\begin{proof}
Otherwise, $xa$ is not supermaximal, because there exists $c \in L_{xa} \cap L_{xb}$, for which $cxa \in \fact_w$ and $cxb\in \fact_w$, and thus $xa$ is a proper suffix of the right-extension $cxa$ of the right-maximal string $cx$. 
\end{proof}

Let us introduce a special tree of supermaximal extensions, which will be helpful to prove our result.

\begin{definition}The \emph{prefix supermaximal right-extension tree} (PSR tree) of $w$ is a compacted trie of all the supermaximal right-extensions of $w$. Since some of those strings can be prefixes of others, internal nodes of the trie may also represent strings, and their corresponding nodes are not compacted. Distinct leaves do not represent prefixes of each other.
\end{definition}

The following result is key to our proof.

\begin{lemma}\label{le:reverse_right_maximal}
Let $w$ be a string starting with a unique symbol $\dol$. If $x$ is a right-maximal substring of $w$, and $xa$ is a supermaximal right-extension of $x$ in $w$, then $x^R$ is a right-maximal substring of $w^R$.
\end{lemma}

\begin{proof}Because $xa$ is a right-extension in $w$, there exists $xb \in \fact_w$ with $b\neq a$. Any symbol preceding $xa$ cannot precede $xb$, otherwise, $xa$ would not be supermaximal. Since $w$ starts with the unique symbol $\dol$, if $x\neq \varepsilon$, then $xa$ and $xb$ have occurrences that are not prefixes of $w$, and hence, they are preceded by different symbols. It follows $x^R$ is right-maximal in $w^R$.
\end{proof}

\begin{lemma}\label{le:reverse_multiplicative}
It always holds that $\sre(w^R)/2\le \sre(w) \le 2\cdot \sre(w^R)$.
\end{lemma}

\begin{proof}
We define an injection $\lambda$ from the set of leaves and unary internal nodes (those with exactly one child) of the PSR tree $T(w)$ of $w$ into $\sreset(w^R)$. We will identify tree nodes with the string they represent. 

\begin{itemize}
\item \textbf{Leaf nodes:} If $xa$ is (the supermaximal right-extension represented by) a leaf node, we let $\lambda(xa)=ux^R\cdot l_{xa}$ where $l_{xa} \in L_{xa}(w)$ (which exists because $w$ starts with $\dol$, as discussed), and $u$ is chosen such that $ux^R\cdot l_{xa} \in \sreset(w^R)$. At least one such $u$ must exist because $x^R \cdot l_{xa}\in \fact_{w^R}$ and $x^R$ is right-maximal in $w^R$ by Lemma \ref{le:reverse_right_maximal}. 

\item \textbf{Unary internal nodes:} If $xa$ a unary internal node with child $yb = xazb$, consider a right-extension $xazc$ of $xaz$, where $c \neq b$ (one must exist because $xaz$ is right-maximal). We then let $\lambda(xa)=uz^Rax^R\cdot l_{yc} \in \sreset(w^R)$ where $l_{yc} \in L_{yc}(w) \subseteq L_{xa}(w)$. Note $L_{yc}$ is not empty because, as in the previous point, $xazc$ is a supermaximal right-extension, and $u$ exists because $z^R a x^R$ is right-maximal in $w^R$ by Lemma \ref{le:reverse_right_maximal}. 
\end{itemize}

We now prove that $\lambda$ is an injection. Consider two leaves $xa$ and $yb$ of $T(w)$:

\begin{enumerate}
\item If $x[i] \neq y[i]$ for some $i$, then $\lambda(xa) = u x^R \cdot l_{xa}$ and $\lambda(yb) = v y^R \cdot l_{yb}$, aligned from the right, differ at their position $\lambda(xa)[|\lambda(xa)|-i] \neq \lambda(yb)[|\lambda(yb)|-i]$, and therefore are different.
\item If $x$ is a proper prefix of $y$, then $y$ has to start with $xc$ for some $c\neq a$ (otherwise $xa$ would be internal). Because $xa$ is a supermaximal right-extension, the set of left-extensions $L_{xa}$ and $L_{yb}\subseteq L_{xc}$ are disjoint by Observation~\ref{obs:Ldisjoint}. Therefore, the selected supermaximal extensions are of the form $\lambda(xa) = ux^R\cdot l_{xa}$ and $\lambda(yb) = vy^R \cdot l_{xc}$ with $l_{xa} \neq l_{xc}$, thus, distinct.
\item Similarly, if $x = y$, the two supermaximal extensions $xa$ and $xb$ have disjoint sets of left-extensions $L_{xa}$ and $L_{xb}$, and hence $\lambda(xa)\neq\lambda(xb)$ because they end in different characters. 
\end{enumerate}

For a graphic depiction of these three cases see Figure \ref{fig:cases_reverse_multiplicative}. 

In the other cases, one of the two nodes is a (unary) internal node. Assume, without loss of generality, that $|x| \le |y|$. By the same argument of point 1 above, $\lambda(xa)$ cannot be equal to some $\lambda(yb)$ if $x$ is not a prefix of $y$. It follows that $xa$ is a unary internal PSR node that is an ancestor of $yb=xazb$:

\begin{enumerate}
\item[4.] Let $yb$ be the child of $xa$. Let $c\neq b$ be such that $xazc$ is a right-extension in $w$, and $\lambda(xa)=uz^Rax^R \cdot l_{yc}$ with $l_{yc} \in L_{yc}$. 
Note that $\lambda(xa)$ differs from $\lambda(yb)$ 
because $\lambda(yb)$ finishes with some $l_{yb} \in L_{yb}$ and, by Observation~\ref{obs:Ldisjoint}, $L_{yb} \cap L_{yc} = \emptyset$. 
\item[5.] 
If $yb$ is a farther descendant of $xa$, then $xaz$ has children by $c$ and by some $d \neq c$ where $xazd$ prefixes $y$. The string $xazd$ is a supermaximal right-extension because there is a node for it in the tree, hence, by Observation \ref{obs:Ldisjoint}, it holds $L_{xazc} \cap L_{xazd} = \emptyset$. As $L_{yb} \subseteq L_{xazd}$, it follows $L_{xazc} \cap L_{yb} = \emptyset$. Therefore, $\lambda(xa)$ and $\lambda(yb)$ end in different symbols.
\end{enumerate}

For a graphic depiction of these two cases see Figure \ref{fig:cases_reverse_multiplicative2}. 

We are then injecting $k$ nodes into $\sreset(w^R)$, where $k \ge \sre(w)/2$ (with equality when $T$ is a full binary tree). Thus, $\sre(w)/2 \le \sre(w^R)$. Since the reversal operation is an involution, we have similarly that $\sre(w^R)/2 \le \sre(w)$. 
\end{proof}

\begin{figure}[t]\center
\begin{tikzpicture}[scale=1]
\node at (-6.5,1) {1)};
\draw[decorate,decoration={brace,amplitude=5pt}] (-6,0.7) -- (-4,0.7)
  node[midway,above,yshift=4pt] {$x$};
\draw (-6.5,0) rectangle (-6,0.5) node[pos=.5,yshift=-0.75pt] {$l_{xa}$};
\draw (-6,0) rectangle (-5.5,0.5) node[pos=.5,yshift=-0.75pt] {$z$};
\draw[fill=gray!40]   (-5.5,0) rectangle (-5,0.5) node[pos=.5,yshift=-0.75pt] {$c$};
\draw (-5,0) rectangle (-4,0.5) node[pos=.5,yshift=-0.75pt] {$x'$};
\draw (-4,0) rectangle (-3.5,0.5) node[pos=.5,yshift=-0.75pt] {$a$};

\draw (-6.5,-0.75) rectangle (-6,-0.25) node[pos=.5,yshift=-0.75pt] {$l_{yb}$};
\draw (-6,-0.75) rectangle (-5.5,-0.25) node[pos=.5,yshift=-0.75pt] {$z$};
\draw[fill=gray!40]   (-5.5,-0.75) rectangle (-5,-0.25) node[pos=.5,yshift=-0.75pt] {$d$};
\draw (-5,-0.75) rectangle (-2.5,-0.25) node[pos=.5,yshift=-0.75pt] {$y'$};
\draw (-2.5,-0.75) rectangle (-2,-0.25) node[pos=.5,yshift=-0.75pt] {$b$};
\draw[decorate,decoration={brace,amplitude=5pt}] (-2.5,-0.9) --  (-6,-0.9) 
  node[midway,below,yshift=-4pt] {$y$};

\node at (-1.5,1) {2)};
\draw[fill=gray!40]   (-1.5,0) rectangle (-1,0.5) node[pos=.5,yshift=-0.75pt] {$l_{xa}$};
\draw (-1,0) rectangle (0,0.5) node[pos=.5,yshift=-0.75pt] {$x$};
\draw (0,0) rectangle (0.5,0.5) node[pos=.5,yshift=-0.75pt] {$a$};

\draw[fill=gray!40]   (-1.5,-0.75) rectangle (-1,0.-0.25) node[pos=.5,yshift=-0.75pt] {$l_{xc}$};
\draw (-1,-0.75) rectangle (0,-0.25) node[pos=.5,yshift=-0.75pt] {$x$};
\draw (0,-0.75) rectangle (0.5,-0.25) node[pos=.5,yshift=-0.75pt] {$c$};
\draw (0.5,-0.75) rectangle (1.5,-0.25) node[pos=.5,yshift=-0.75pt] {$y'$};
\draw (1.5,-0.75) rectangle (2,-0.25) node[pos=.5,yshift=-0.75pt] {$b$};
\draw[decorate,decoration={brace,amplitude=5pt}] (1.5,-0.9) --  (-1,-0.9) 
  node[midway,below,yshift=-4pt] {$y$};

\node at (2.5,1) {3)};
\draw[fill=gray!40]   (2.5,0) rectangle (3,0.5) node[pos=.5,yshift=-0.75pt] {$l_{xa}$};
\draw (3,0) rectangle (4,0.5) node[pos=.5,yshift=-0.75pt] {$x$};
\draw (4,0) rectangle (4.5,0.5) node[pos=.5,yshift=-0.75pt] {$a$};

\draw[fill=gray!40]   (2.5,-0.25) rectangle (3,-0.75) node[pos=.5,yshift=-0.75pt] {$l_{xb}$};
\draw (3,-0.25) rectangle (4,-0.75) node[pos=.5,yshift=-0.75pt] {$x$};
\draw (4,-0.25) rectangle (4.5,-0.75) node[pos=.5,yshift=-0.75pt] {$b$};
\draw[decorate,decoration={brace,amplitude=5pt}] (4,-0.9) --  (3,-0.9) 
  node[midway,below,yshift=-4pt] {$y$};

\end{tikzpicture}
\caption{Cases 1), 2), and 3) of Lemma \ref{le:reverse_multiplicative}, where $xa$ and $yb$ are leaves of the PSR tree. Note that the suffix $x^R\cdot l_{xa}$ and $x^R\cdot l_{yb}$ of $\lambda(xa)$ and $\lambda(yb)$ are not a suffix one of the other because the grayed symbols are different.
}\label{fig:cases_reverse_multiplicative}
\end{figure}

\begin{figure}[t]\center
\begin{tikzpicture}[scale=1]
\draw (0,0) rectangle (1,0.5) node[pos=.5,yshift=-0.75pt] {$x$};
\draw (1,0) rectangle (1.5,0.5) node[pos=.5,yshift=-0.75pt] {$a$};

\draw[thick,->] (1.75,0.25) to (2.75,0.25);
\draw[dashed,->] (1.75,0.25) to (2.75,-0.75);
\draw[thick,->] (6.75,0.25) to (7.75,0.25);
\draw (7.25,0.5) node {$\cdots$};

\draw[decorate,decoration={brace,amplitude=5pt}] (3.5,0.7) -- (6,0.7)
  node[midway,above,yshift=4pt] {$y$};
\draw[fill=gray!40]   (3,0) rectangle (3.5,0.5) node[pos=.5,yshift=-0.75pt] {$l_{yb}$};
\draw (3.5,0) rectangle (4.5,0.5) node[pos=.5,yshift=-0.75pt] {$x$};
\draw (4.5,0) rectangle (5,0.5) node[pos=.5,yshift=-0.75pt] {$a$};
\draw (5,0) rectangle (6,0.5) node[pos=.5,yshift=-0.75pt] {$z$};
\draw (6,0) rectangle (6.5,0.5) node[pos=.5,yshift=-0.75pt] {$b$};

\draw[fill=gray!40]  (3,-0.5) rectangle (3.5,-1) node[pos=.5,yshift=-0.75pt] {$l_{yc}$};
\draw (3.5,-0.5) rectangle (4.5,-1) node[pos=.5,yshift=-0.75pt] {$x$};
\draw (4.5,-0.5) rectangle (5,-1) node[pos=.5,yshift=-0.75pt] {$a$};
\draw (5,-0.5) rectangle (6,-1) node[pos=.5,yshift=-0.75pt] {$z$};
\draw (6,-0.5) rectangle (6.5,-1) node[pos=.5,yshift=-0.75pt] {$c$};

\draw[decorate,decoration={brace,amplitude=5pt}] (8.5,0.7) -- (11,0.7)
  node[midway,above,yshift=4pt] {$y'$};
\draw[fill=gray!40]   (8,0) rectangle (8.5,0.5) node[pos=.5,yshift=-0.75pt] {$l_{y'd}$};
\draw  (8.5,0) rectangle (9.5,0.5) node[pos=.5,yshift=-0.75pt] {$y$};
\draw  (9.5,0) rectangle (10,0.5) node[pos=.5,yshift=-0.75pt] {$b$};
\draw  (10,0) rectangle (11,0.5) node[pos=.5,yshift=-0.75pt] {$\cdots$};
\draw  (11,0) rectangle (11.5,0.5) node[pos=.5,yshift=-0.75pt] {$d$};

\end{tikzpicture}
\caption{Case 4) and 5) of Lemma \ref{le:reverse_multiplicative}, where $xa$ is an internal node of the PSR tree, $yb$ is the only child of $xa$, and $y'd$ is a farther descendant of $xa$ (if it exists). Note $xazc$ is a right-extension, but not supermaximal (otherwise, $xa$ would not be unary). In this case, $\lambda(xa)$ is suffixed by $y^R\cdot l_{yc}$, whereas $\lambda(yb)$ and $\lambda(y'd)$ are suffixed by $y^R\cdot l_{yb}$ and $y^R\cdot l_{y'd}$, respectively. Note how $yb$ and all its descendants yield $\lambda(\cdot)$ values differing from $\lambda(xa)$ in their last symbol.}\label{fig:cases_reverse_multiplicative2}
\end{figure}

\subsection{Sensitivity as a function of \texorpdfstring{$\delta$}{delta}}

Finally, we give upper bounds on the sensitivity of $\chi$ to string operations in terms of measure $\delta$.

\begin{lemma}Let $w\in \Sigma^*$  and $w' \in \mathtt{ins}(w) \cup \mathtt{del}(w) \cup \mathtt{sub}(w) \cup \mathtt{rot}(w) \cup \mathtt{rev}(w)$. For $\delta=\delta(w)$, it holds
\begin{align*}
&\chi(w') - \chi(w) = O\left(\delta\max\left(1, \log(n/(\delta\log \delta))\right)\log \delta\right) \text{ and } \\ 
&\chi(w')\hspace{3pt}/\hspace{4pt}\chi(w) = O\left(\max\left(1, \log(n/(\delta\log \delta))\right)\log \delta\right).
\end{align*}
\end{lemma}
\begin{proof}
It holds $\delta\leq\chi\leq 2\overline{r}$~\cite{suffixientarrays} and \sloppy$r=O(\delta\max(1, \log(n/(\delta\log \delta)))\log \delta)$~\cite{r_delta_bound}. Since the multiplicative sensitivity of $\delta$ to any of the string operations is $O(1)$~\cite{AFI23}, for any $w\in\Sigma^*$ it holds $\overline{r}(w)=r(w^R)=O(\delta\max(1, \log(n/(\delta\log \delta)))\log \delta)$.
The thesis follows by considering the worst case, that is, $\chi(w)=\Theta(\delta)$ and \sloppy$\chi(w')=\Theta(\delta\max(1, \log(n/(\delta\log \delta)))\log \delta)$.
\end{proof}

\section{Relating \texorpdfstring{$\chi$}{chi} to Other Repetitiveness Measures}

Previous work~\cite{DGLMP23} established that $\gamma \le \chi \le 2\overline{r}$ on every string family.
In this section we obtain the more natural upper bound $\chi \le 2r$, and that $\chi$ can be asymptotically strictly smaller, $\chi=o(r)$, on some string families (we actually prove $\chi=o(v)$). We also show that $\chi$ is incomparable with all the copy-paste measures except $s$ and $b$, in the sense that there are string families where $\chi$ is asymptotically strictly smaller than the others, and vice versa. We recall that $s = O(\chi)$ and that the relation between $b$ and $\chi$ remains unknown.

\subsection{Proving \texorpdfstring{$\chi \le 2r$}{chi is upper-bounded by r}}

We first prove that $\chi$ is asymptotically upper-bounded by the number $r$ of runs in the BWT of the sequence. 
As for the measure $\chi$, we assume that the BWT is computed after appending the $\dol$ symbol. 
Observe that the precise bound of the following lemma has been recently proved to be tight for binary strings, while on $\sigma$-ary alphabets, with $\sigma>2$, there exists a family of strings that approaches it as $\sigma$ grows~\cite{chi2r}.

\begin{lemma}\label{lem:chi_bounded_by_r}
It always holds that $\chi \le 2r$.
\end{lemma}

\begin{proof}
    Let $x_i$ denotes the $i$th rotation of $w\dol$ in lexicographic order, for each $i\in [1\dd|w|+1]$, and let $u_i$ be the longest common prefix between the rotations $x_i,x_{i+1}$, for each $i\in [1\dd|w|]$ (note this implies that $u_i$ is right-maximal).
    We further define $s:[1\dd n+1]\rightarrow[0\dd n]$ as $s(i) = j$ if $x_i = w[j+1\dd |w|]\dol w[1\dd j]$, that is, the number of cyclic shifts to the right required to transform $x_i$ into $w\dol$.\footnote{The function $s$ mimics the well-known Suffix Array~\cite{ManberM93}, here omitted for simplicity of exposition.} As the symbol $\dol$ occurs only once in $w\dol$, the function $s$ is bijective.
  
    Note that each right-extension of $w\dol$ can be written as $u_ic$, for some $i \in [1\dd |w|]$ and $c\in \Sigma$.
    Consider now the set \[S = \bigcup_{i\in[1\dd |w|]}\{s(i)+|u_i|+1,s(i+1)+|u_i|+1\},\] that is, the set of positions where the occurrences of the right-extensions $u_ic_1$ and $u_ic_2$ end in $w\dol$, where $u_ic_1$ and $u_ic_2$ are the prefixes of $x_i$ and $x_{i+1}$, respectively, for some $c_1,c_2\in\Sigma$ such that $c_1<c_2$. It follows by construction that the set $S$ is a suffixient set of $w\dol$.
    
    We now show that $|S|\leq2r$.
    Let us factorize each pair of consecutive rotations in the BWT-matrix as $x_i = u_iv_ic_i$ and $x_{i+1} = u_iv'_ic_{i+1}$. Observe that $v_i,v'_i\neq \varepsilon$~\cite[Corollary~8]{FRSU23}, $v_i[1] \neq v'_i[1]$, and $c_i = \bwt(w\dol)[i]$ for all $i\in [1\dd |w|+1]$.
    A well-known property of the BWT-matrix is that if $c_i = c_{i+1} = c\in \Sigma$, then there exists $j\in [1\dd |w|]$ such that $x_j = c u_i v_i$ and $x_{j+1} = c u_i v'_i$~\cite{BW94}.
    As a consequence, one has that $s(j)+|u_j|+1 = (s(i)-1) + (|u_i|+1) + 1 = s(i) + |u_i| + 1$ and $s(j+1)+|u_j|+1 = (s(i+1)-1) + (|u_i|+1) + 1 = s(i+1) + |u_i| + 1$, and the procedure can be reiterated as long as $x_j$ and $x_{j+1}$ end with the same symbol. See Figure \ref{fig:bwt} for an illustration.
    It follows that the same set can be written as \[S=\{s(i)+|u_i|+1,s(i+1)+|u_i|+1\mid i\in[1\dd |w|] \wedge\bwt[i]\neq\bwt[i+1]\},\] that is, the size of $S$ is at most twice the number of equal-letter runs in $\bwt(w\dol)$, and the thesis follows.
\end{proof}

\begin{figure}\center
\begin{tikzpicture}[scale=1]

\draw[-] (0,2) to (0,9);
\draw[-] (9.5,2) to (9.5,9);
\draw[-] (10,2) to (10,9);
\node at (9.75,9.5) {$\mathtt{BWT}(w\dol)$};

\node at (-1,8.25) {$x_j$};
\node at (-1,7.70) {$x_{j+1}$};
\draw (0,8) rectangle (3.5,8.5) node[pos=.5] {$u$};
\draw (0,7.5) rectangle (3.5,8) node[pos=.5] {$u$};
\draw (3.5,8) rectangle (4,8.5) node[pos=.5] {$a$};
\draw (3.5,7.5) rectangle (4,8) node[pos=.5] {$b$};
\draw (9.5,8) rectangle (10,8.5) node[pos=.5] {$c$};
\draw (9.5,7.5) rectangle (10,8) node[pos=.5] {$c$};
\draw[->] (-1.45,8) --(-1.6,8) --(-1.6,5.5) --(-1.45,5.5);

\node at (9.75,6.75) {$\vdots$};

\node at (-1,5.75) {$x_{j'}$};
\node at (-1,5.20) {$x_{j'+1}$};
\draw (0.5,5.5) rectangle (4,6) node[pos=.5] {$u$};
\draw (0.5,5) rectangle (4,5.5) node[pos=.5] {$u$};
\draw (4,5.5) rectangle (4.5,6) node[pos=.5] {$a$};
\draw (4,5) rectangle (4.5,5.5) node[pos=.5] {$b$};
\draw (0,5.5) rectangle (0.5,6) node[pos=.5] {$c$};
\draw (0,5) rectangle (0.5,5.5) node[pos=.5] {$c$};
\draw (9.5,5.5) rectangle (10,6) node[pos=.5] {$c$};
\draw (9.5,5) rectangle (10,5.5) node[pos=.5] {$c$};
\draw[decorate,decoration={brace,amplitude=10pt}]  (0,6) -- (4,6) node[midway,above,yshift=10pt] {$u'$};
\draw[->] (-1.75,5.5) --(-1.9,5.5) --(-1.9,3) --(-1.45,3);

\node at (9.75,4.25) {$\vdots$};

\node at (-1,3.25) {$x_{j''}$};
\node at (-1,2.70) {$x_{j''+1}$};
\draw (1,3) rectangle (4.5,3.5) node[pos=.5] {$u$};
\draw (1,2.5) rectangle (4.5,3) node[pos=.5] {$u$};
\draw (4.5,3) rectangle (5,3.5) node[pos=.5] {$a$};
\draw (4.5,2.5) rectangle (5,3) node[pos=.5] {$b$};
\draw (0,3) rectangle (0.5,3.5) node[pos=.5] {$c$};
\draw (0,2.5) rectangle (0.5,3) node[pos=.5] {$c$};
\draw (0.5,3) rectangle (1,3.5) node[pos=.5] {$c$};
\draw (0.5,2.5) rectangle (1,3) node[pos=.5] {$c$};
\draw (9.5,3) rectangle (10,3.5) node[pos=.5] {$d$};
\draw (9.5,2.5) rectangle (10,3) node[pos=.5] {$e$};
\draw[decorate,decoration={brace,amplitude=10pt}]  (0,3.5) -- (5,3.5) node[midway,above,yshift=10pt] {$u''$};

\end{tikzpicture}
\caption{Depiction of the proof of Lemma \ref{lem:chi_bounded_by_r}.}\label{fig:bwt}
\end{figure}

In some fields like Combinatorics on Words, the BWT is used without appending the dollar symbol at the end of the string.  We call this variant the  \emph{circular BWT} and denote its number of runs by $r_c$. This small change can have a great impact in the value of the measure:  $r$ and $r_c$ can differ by a $\Omega(\log n)$ factor \cite{GILRSU25}.   We show that $\chi$ is also smaller than the variant $r_c$.

\begin{lemma}It always holds that $\chi \le 2r_c+2$.
\end{lemma}

\begin{proof}
Consider this time the BWT-matrix of $w$ instead of $w\dol$. Observe that all the right-extensions of $w\dol$ still appear as $u_ia$ for some $u_i=lcp(x_i,x_{i+1})$ and symbol $a \in \Sigma$, with some exceptions:

\begin{enumerate}
\item \label{case1} the right-extensions ending with $\dol$, and
\item \label{case2}the right-extensions $za$ ending with some symbol $a$ such that only $za$ and $z\dol$ are right-extensions.
\end{enumerate}

The right-extensions in Case~\ref{case1} and Case~\ref{case2} form two suffix chains, and hence,  the set  \[S = \bigcup_{i\in[1\dd |w|-1]}\{s(i)+|u_i|+1,s(i+1)+|u_i|+1\},\] is missing positions for at most 2 supermaximal right-extensions. Thus, $\chi \le 2r_c+2$.
\end{proof}

\subsection{A family with \texorpdfstring{$\chi=o(v)$ (and thus $o(r)$)}{chi smaller than v}}

We will now show that $\chi = o(v)$ on the so-called Fibonacci words, which also implies $\chi=o(r)$ in that string family because $v=O(r)$~\cite{NOP21}. Combined with Lemma~\ref{lem:chi_bounded_by_r}, this implies that $\chi$ is a strictly smaller measure than $r$. In contrast, $\chi$ is incomparable with $v$, as we show later. On our way, we obtain some relevant byproducts about the structure of suffixient sets on Fibonacci, and more generally, episturmian words.

\begin{definition}[\cite{epistand_episturm,episturmian_survey}]An infinite string $\infw$ is \emph{episturmian} if it has at most one right-maximal substring of each length and its set of substrings is closed under reversal, that is, $\fact_\infw = \fact_\infw^R$. It is \emph{standard episturmian} (or \emph{epistandard}) if, in addition, all the right-maximal substrings of $\infw$ are of the form $\infw[1\dd i]^R$ with $i \ge 0$, that is, they are the reverse of some prefix of $\infw$.
\end{definition}

\begin{lemma}
Let $\textbf{w}\in \Sigma^\omega$ be an episturmian word with $|\Sigma| = \sigma \ge 2$. Then, \sloppy $\sre(\textbf{w}[i\dd j]) \le \sigma$ for any $i,j\ge 0$.
\end{lemma}

\begin{proof}
Let $\infw$ be an epistandard word.  The right-extensions $x_1,x_2 , \dots $ ending with $a\in \Sigma$ form a \emph{suffix-chain} where each $x_i$ is a suffix of $x_{i+1}$. This is because the right-maximal substrings are the reverses of the prefixes of the word, so they are all suffixes of the longest one. If we add the same letter to a subset of them to obtain right-extensions, they still form a suffix-chain. There is one of those suffix-chains for each character $a \in \Sigma$.

Let $\infw$ be episturmian but not necessarily epistandard. There exists some epistandard word $\textbf{s}$ with the same set of substrings, i.e., $\fact_\infw = \fact_{\textbf{s}}$~\cite{epistand_episturm}. Therefore, for any episturmian word $\infw$, there exist exactly $\sigma$ suffix-chains of right-extensions.

 When considering substrings of $\infw$, the supermaximal right-extension in $\infw[i\dd j]$ ending with $a \in \Sigma$ is the longest right-extension of $\textbf{w}$ ending with $a$ that remains a right-extension in $\infw[i\dd j]$. It follows that for any substring $\infw[i\dd j]$ of any episturmian word $\infw$, it holds $\sre(\infw[i\dd j]) \le \sigma$.  \end{proof}

Combining this result with Corollary~\ref{cor:sre_chi}, we obtain the following bound.

\begin{corollary} \label{cor:morphic}
For any episturmian word $\infw \in \Sigma^\omega$ it holds $\chi(\infw[i\dd j]) \le \sigma + 2$. \end{corollary}

The next lemma precisely characterizes the suffixient sets of Fibonacci words, a particular case of epistandard words that will be useful to relate $\chi$ with $v$.

\begin{definition}
Let $F_1 = \b$, $F_2 = \a$, and $F_k = F_{k-1}F_{k-2}$ for $k \ge 3$ be the Fibonacci family of strings. Their lengths, $f_k = |F_k|$, form the Fibonacci sequence.  \end{definition}

\begin{lemma} \label{lem:fibo}
Every Fibonacci word $F_k\dol$ has a suffixient set of size at most 4. For $k \ge 6$, the only smallest suffixient sets for $F_k\dol$ are $\{ f_k+1, f_k-1, f_{k-1}-1, p \}$, where $p \in \{f_{k-2}+1,2f_{k-2}+1\}$.
\end{lemma}
\begin{proof}
The upper bound of $4$ stems directly from Corollary~\ref{cor:morphic}, because the infinite Fibonacci word is binary epistandard.
For $k \ge 3$, there exist strings $H_k$ such that $F_k = F_{k-1}F_{k-2} = H_kcd$ and $F_{k-2}F_{k-1}=H_kdc$, for $cd = \a\b$ or $cd=\b\a$ depending on the parity of $k$~\cite{deLuca81}.
Let us call $F'_k = H_k dc = F_{k-2}F_{k-1}$, that is, $F_k$ with the last two letters exchanged;
thus $F_k = F_{k-1}F_{k-2} = F_{k-2}F_{k-1}'$.

Note that $F_{k-1} = H_{k-1} dc$ prefixes $F_k$. On the other hand,
we can write $F_k = F_{k-1}F_{k-2} = F_{k-2}F_{k-3}F_{k-2} = F_{k-2}F_{k-1}'=F_{k-2}H_{k-1}cd$.
Therefore, string $H_{k-1}$ is right-maximal in $F_k$. Its extensions, $H_{k-1}d$ and $H_{k-1}c$, are supermaximal because there are no other occurrences of $H_{k-1}$ in $F_k$: (i) $H_{k-1}$ cannot occur starting at positions $f_{k-2}+2$ or $f_{k-2}+3$ because it occurs at $f_{k-2}+1$, so $H_{k-1}$ should match itself with an offset of 1 or 2, which is impossible because it prefixes $F_{k-1}$ and all $F_{k-1}$ for $k-1 \ge 5$ start with $\str{abaab}$; (ii) $H_{k-1}$ cannot occur starting at positions $2$ to $f_{k-2}$ because its prefix $F_{k-2}$ should occur inside the prefix $F_{k-2}F_{k-2}$ of $F_k = F_{k-2}F_{k-1}' = F_{k-2}F_{k-2}F_{k-3}'$, and so $F_{k-2}$ should equal a rotation of itself, which is impossible~\cite[Cor.~3.2]{Dro95}. The two positions following $H_{k-1}$, $f_{k-1}-1$ and $f_k-1$, then appear in any suffixient set.

On the other hand, $F_{k-2}$ is followed by $\dol$ in $F_k\dol$, and it also prefixes $F_k = F_{k-2}F_{k-1}'$, therefore $F_{k-2}$ is right-maximal. The first occurrence is preceded by $F_{k-1}$, and hence by $c$, and the second by no symbol. $F_{k-2}$ also occurs in $F_k$ at position $f_{k-2}+1$, as seen above, preceded by $F_{k-2}$ and thus by $d$. There are no other occurrences of $F_{k-2}$ in $F_k$ because (i) it cannot occur starting at positions $2$ to $f_{k-2}$ by the same reason as point (ii) of the  previous paragraph; (ii) it cannot appear starting at positions $f_{k-2}+2$ to $f_{k-1}-2$ because $F_k = F_{k-2}F_{k-2}F_{k-3}'$ and $F_{k-3}'[1,f_{k-3}-2] = F_{k-3}[1,f_{k-3}-2] = F_{k-2}[1..f_{k-3}-2]$, thus such an occurrence would also match a rotation of $F_{k-2}$, which is impossible as noted above; (iii) it cannot appear starting at positions $f_{k-1}-1$ or $f_{k-1}$ because, since it matches at position $f_{k-1}+1$, $F_{k-2}$ would match itself with an offset of 1 or 2, which is impossible as noted in point (i) of the previous paragraph. The right-extensions of $F_{k-2}$ are then supermaximal. The one followed by $\dol$ occurs ending at position $f_k+1$. The other two are followed by $\a$ because they are followed by $F_{k-2}$ and by $F_{k-3}'$ and all $F_k$ for $k \ge 2$ start with $\a$. We can then choose either ending position for a suffixient set, $f_{k-2}+1$ or $2f_{k-2}+1$.

\end{proof}

\begin{corollary} \label{cor:v}
There exist string families where $\chi = o(v)$.
\end{corollary}
\begin{proof}
It follows from Lemma~\ref{lem:fibo} and the fact that $v=\Omega(\log n)$ on the odd Fibonacci words~\cite[Thm.~28]{NOP21}. 
\end{proof}

\subsection{Uncomparability of \texorpdfstring{$\chi$}{chi} with copy-paste measures}

Finally, we show that $\chi$ is incomparable with most copy-paste measures. This follows  from $\chi$ being $\Theta(n)$ on de Bruijn sequences and $O(1)$ on Fibonacci strings.

\begin{definition}A  \emph{binary de Bruijn sequence of order $k>0$}~\cite{de_bruijn} contains every binary string in $\{\a,\b\}^k$ as a substring exactly once. The length of this sequence is $n = 2^k+(k-1)$. The set of binary de Bruijn sequences of order $k$ is $\db(k)$.\end{definition}

\begin{lemma}\label{lem:chi_de_bruijn}
It holds $\sre(w) = 2^k = \Omega(n)$ for any $w[1 \dd n] \in \db(k)$.
\end{lemma}

\begin{proof}

Let $w[1\dd n]$ be a binary de Bruijn string of order $k$. By definition, $w$ contains every binary string of length $k$ as a substring exactly once. As all the possible pairs of strings $x\a$ and $x\b$ of length $k$ appear in  $w$, it follows that all the strings in $\mathcal{F}_w(k)$ are right-extensions. 
Moreover, each $x\a$ and $x\b$ of length $k$ are supermaximal right-extensions: otherwise, there would exist some $c\in\{\a,\b\}$ such that $cx\a$ and $cx\b$ are both substrings of $w$, which raises a contradiction since the $k$-length string $cx$ cannot appear twice in $w$.
Moreover, there are no right-maximal strings of length $k$ or greater; hence, there are no right-extensions of length greater than $k$. It follows that $\sre(w) =  |\fact_w(k)| = 2^k = \Omega(n)$.  \end{proof} 

Because $g = O(n/\log n)$ on de Bruijn sequences~\cite{NOP21} and by Lemma~\ref{lem:chi_de_bruijn}, we have: 

\begin{corollary}\label{cor:chi_uncomparable_g}There exists a string family with $\chi=\Omega(g\log n)$.\end{corollary}

This result is particularly relevant because all the copy-paste based measures $\mu$, with the exception of $z_e$, are $O(g)$. Corollary \ref{cor:chi_uncomparable_g} then implies $\mu=o(\chi)$ on de Bruijn sequences for all these measures $\mu$.

While it has been said that $z_e = O(n/\log n)$ on binary sequences as well~\cite{KN13}, this referred to the version that adds to each phrase the next nonmatching character. Because $z_e$ is not an optimal parse, it is not obvious that this also holds for the version studied later in the literature, which does not add the next character. We then prove next that $z_e = o(\chi)$ holds on de Bruijn words.

\begin{lemma} \label{lem:ze}
There exists a string family with $\chi = \Omega\left(z_e \frac{\log n \log\log\log n}{(\log\log n)^2}\right)$.
\end{lemma}
\begin{proof}
It always holds that $z_e = O\left(z \frac{\log^2(n/z)}{\log\log(n/z)}\right)$~\cite{GKM23}. In de Bruijn sequences it holds that $z=\Theta(n/\log n)$, so $n/z=\Theta(\log n)$. Therefore, $z_e=O\left(z \frac{(\log\log n)^2}{\log\log\log n}\right)$, and replacing $z=\Theta(n/\log n)$ we get $z_e = O\left(n\frac{(\log\log n)^2}{\log n\log\log\log n}\right)$. By Lemma~\ref{lem:chi_de_bruijn}, this yields $\chi=\Omega\left(z_e \frac{\log n \log\log\log n}{(\log\log n)^2}\right) = \omega(z_e)$ on de Bruijn sequences.
\end{proof}

\begin{corollary}\label{cor:chi_uncomparable_to_other_measures}
The measure $\chi$ is uncomparable to $\mu \in \{z,z_{no},z_e,z_{end},v,g,g_{rl},c\}$.
\end{corollary}
\begin{proof}
From Corollary~\ref{cor:chi_uncomparable_g} and Lemma~\ref{lem:ze}, and that  $z$, $z_{no}$, $z_{end}$, $v$, $g_{rl}$ and $c$ are always $O(g)$, it follows that there are string families where $\mu = o(\chi)$, for any $\mu \in \{z,z_{no},z_e,z_{end},v,g,g_{rl},c\}$. On the other hand, from Lemma~\ref{lem:fibo} and Corollary~\ref{cor:v}, and that $c = \Omega(\log n)$ on Fibonacci words~\cite[Thm.~32]{NOP21} and $c = O(\mu)$ for any $\mu \in \{z,z_{no},z_e,z_{end},g_{rl},g\}$~\cite[Thm.~30]{NOP21}, it follows that there are string families where $\chi = o(\mu)$, for any $\mu \in \{z,z_{no},z_e,z_{end},v,g,g_{rl},c\}$.
\end{proof}

\section{Online Computation of Smallest Suffixient Sets} \label{sec:ukkonen}

As an application of Lemmas~\ref{lem:sre_append_bounded} and \ref{lem:sre_prepend_bounded}, we show that Ukkonen's \cite{Ukk95} and Weiner's \cite{Wei73} linear-time online constructions of suffix trees can be easily modified to compute smallest suffixient sets within the same space and time complexity.

The {\em suffix tree} of a text $T$ is a tree of size $O(|T|)$ where edges are labeled by nonempty substrings of $T$ (the label $T[i\dd j]$ is indicated by the pair $(i,j)$). Every internal node has at least two children, and the labels of edges to any two children must differ in their first symbol. If node $x$ has a child node $y$ by an edge labeled $(i,j)$, we say that $y$ is a child {\em by label} $T[i]$. Each node $x$ is said to be the {\em locus} of the string obtained by concatenating the labels of the edges that lead from the root to $x$. The {\em string depth} of a node $x$ is the length of the string $x$ is the  locus of. The tree leaves are the loci of suffixes $T[i\dd]$, whereas internal nodes are loci of strings that occur at least twice in $T$. Edges to leaves have labels of the form $(i,\infty)$, which means the second component is always the last position processed of $T$. Suffix tree nodes also have so-called {\em suffix links}, which lead from a node that is the locus of a string $cu$, for $c \in \Sigma$, to the locus of string $u$ (which always exists: if $cu$ occurs more than once, so does $u$). Finally, we extend the concept of suffix tree nodes to {\em implicit nodes}, which are virtual nodes with one child, assumed to exist along edges: if $y$ is the child of $x$ by an edge labeled $(i,j)$ and $x$ is the locus of string $u$, then $(x,(i,p))$, for $i \le p < j$, is an implicit node that is the locus of $u \cdot T[i\dd p]$. Knowing the loci of which string they are, suffix links are also defined from implicit nodes (those can lead to explicit or to implicit nodes).

\subsection{Ukkonen's left-to-right construction}

Ukkonen's algorithm \cite{Ukk95} builds the suffix tree of $T$ such that, after having processed any prefix $w$ of $T$, it has built the suffix tree of $w$. To prevent special cases, it assumes that the root is in fact a child of a special node $\bot$, with edges to the root labeled $c$ for every $c \in \Sigma$; the suffix link of the root points to $\bot$. We do not aim at a full explanation of the algorithm, but just highlight some of its relevant properties. Let the prefix $w$ of $T$ be followed by $c \in \Sigma$. The algorithm maintains pointers to two (possibly implicit) nodes of the suffix tree of $w$:
\begin{description}
\item[$s$:] called the {\em active point} of $w$, is the locus of the longest suffix $u$ of $w$ that occurs at least twice in $w$;
\item[$s'$:] called the {\em end point} of $w$, is the locus of the longest suffix $v$ of $w$ such that $vc$ occurs in $w$.
\end{description}

Note that this implies that the locus of $vc$ is the active point of $wc$ (i.e., the new $s$ after processing $c$), and that, because $v$ must be a suffix of $u$, there is a {\em chain} of consecutive suffix links from $s$ to $s'$. The algorithm updates the suffix tree nodes in that chain, from $s$ to (but not including) $s'$, by adding a new leaf child labeled $c$ to those nodes. Although updating the suffix tree of $w$ to obtain that of $wc$ may take non-constant time, the time does amortize to constant \cite{Ukk95}.

We will enhance the suffix tree nodes by adding a {\em mark} to the nodes that are loci of the supermaximal right-extensions of a smallest suffixient set of $w$. The length of the extension is the string depth of the marked node, and any suffix descending from the node serves as a starting position of such extension. This is the way in which we maintain a smallest suffixient set for the current prefix $w$ of $T$. For example, we can easily maintain the marked nodes in a list in order to collect the set for any prefix in optimal time. To compute $\chi(T)$, we can run the algorithm on $T\dol$ and then return the length of the list. 

Note that the loci of right-extensions are either explicit nodes, or implicit nodes of the form $(x,(i,i))$, because they extend by one symbol a string that occurs more than once. We can then always associate the mark to the corresponding explicit child $y$ of $x$, so as to consult and update it in constant time.

The first letter $c$ of $T$ is processed by adding a child leaf of the root
by label $c$ and setting $s$ to the root. From there on,
considering the cases of Lemma~\ref{lem:sre_append_bounded}, we proceed as follows:
\begin{description}
\item[$s$ has a child by label $c$:] This means that $s=s'$ and Ukkonen's algorithm just descends by $c$ to find the new $s$. We do not need any further action because we are in case 1 of the lemma.
\item[$s$ has two or more children, none by label $c$:] This is case 2 of the lemma. Ukkonen's algorithm will create a new leaf child by label $c$ of $s$, and of all the nodes in the suffix-link chain until it finds $s'$ (i.e., the first node in the chain having a child by label $c$). We then (1) mark the (just created) child of $s$ by label $c$, and (2) unmark, if it is marked, the (already existing) child of $s'$ by label $c$. This is because $uc$ is a new supermaximal right extension of the right-maximal string $u$ and, in case (2), $vc$ ceases to be a supermaximal right extension because it is a suffix of the new one we are adding, $uc$.
\item[$s$ has just one child, by label $a \neq c$:] This is case 3 of the lemma. Ukkonen's algorithm proceeds exactly as in the previous case, finding the first $s'$ with a child by label $c$ in the suffix-link chain. We (1) mark the two children of $s$ (by labels $a$ and $c$), (2) unmark, if it is marked, the child of $s'$ by label $c$, and (3) unmark, if it is marked, the child of $s''$ by label $a$, where $s''$ is the first node in the suffix-link chain that has another child labeled $b \neq a$ (note that $s''$ must be on the chain from $s$ to $s'$, because one possible choice is $b=c$ and $s''=s'$). This is because $ua$ and $uc$ are new supermaximal right-extensions of $u$, and in case (2), $vc$ is not anymore supermaximal, as before. In case (3), similarly, if $s''$ is the locus of $z$, then $za$ is not anymore supermaximal. Note that the child by label $a$ of $s''$ is the only locus of some $za$ suffix of $ua$ that could possibly be marked. 
\end{description}

Ukkonen's algorithm is claimed to be $O(n)$ time, but this assumes that the alphabet is constant. If it is not, we may need more time to find the child by label $c$ among its children, or determine it does not exist. If the alphabet is an integer range and of size polynomial in $n$, we can still retain $O(n)$ time in the transdichotomous RAM model of computation with computer word size $\omega = \Omega(\log n)$ using fusion trees, as the children operations can then be handled in time $O(\log_\omega|\Sigma|)=O(1)$ \cite{PT14}. Otherwise, an extra factor of $O(\log|\Sigma|)$ appears.

\begin{theorem}
There exists an algorithm to compute smallest suffixient sets that processes a text $T$ left to right such that, for every $n$, after having processed the prefix $T[1\dd n]$ it has determined a smallest suffixient set of $T[1\dd n]$ using $O(n)$ space and $O(n)$ worst-case time. This can be used to compute $\chi(T)$ within $O(|T|)$ space and time.
\end{theorem}

\subsection{Weiner's right-to-left construction}

Weiner's algorithm \cite{Wei73} process the text right-to-left, prepending a symbol at each step. Analogously to Ukkonen's algorithm, after having processed a suffix of $T$, it has built the suffix tree of that suffix. The algorithm is possibly less convenient than Ukkonen's, as it is more natural to process the text left to right, but its adaptation to compute $\sre$ using Lemma~\ref{lem:sre_prepend_bounded} needs diving less into its details. Further, it has the property that we immediately know $\chi$ for each suffix of $T$ as soon as we process it.

The algorithm is actually defined for a $\$$-terminated text, $T\$$.
Assume we have already processed the suffix $w\$$ of $T\$$, and now prepend $c$
to obtain the suffix tree of the suffix $cw\$$. The algorithm maintains a pointer to the deepest explicit node $u$ in the path to the leaf that represents $w\$$. To go from $w\$$ to $cw\$$, the algorithm finds, or creates, the deepest explicit node $v$ in the path to $cw\$$, and adds as a child of $v$ a new leaf that represents $cw\$$. 

While the algorithm manages to find $v$ in constant amortized time from $u$ using (and maintaining) the so-called ``Weiner links'' (which are the inverse of the suffix links), let us reason as if we found $v$ by descending from the root with the symbols of $cw\$$. We would descend by the path until we find an explicit or implicit node $v$ that is the locus of a prefix $cx$ of $cw\$$ from where no child descends by label $a$, where $cxa$ also prefixes $cw\$$. Those are the strings $cx$ and $cxa$ alluded to in Lemma~\ref{lem:sre_prepend_bounded}. The algorithm then:
\begin{itemize}
    \item If $v$ is implicit, converts it into an explicit node with (temporarily) a single child, which descends by label, say, $b$.
    \item Creates a new child of $v$ by label $a \neq b$, which is a leaf representing $cw\$$. 
\end{itemize}
By Lemma~\ref{lem:sre_prepend_bounded}, $cx$ is right-maximal in $cw\$$ (this is also clear because the explicit node $v$ is its locus). To obtain $\mathcal{S}_r(cw\$)$, we need to add to $\mathcal{S}_r(w\$)$ at most two right-extensions: $cxa$, which is represented by the child of $v$ by label $a$ and, in case $v$ was originally implicit, $cxb$, the other child of $v$. We maintain $\mathcal{S}_r$ by marking the edges leaving from the corresponding right-maximal strings. Consequently, we now mark the two edges that leave from $v$. On the other hand, if $xa$ and/or $xb$ were in $\mathcal{S}_r(w\$)$, they must be removed to form $\mathcal{S}_r(cw\$)$ as they are now suffixes of $cxa$ and $cxb$. We thus unmark, if marked, the two children of $u$ that descend by $a$ and $b$ (both of which exist).

A suffixient set of string positions can be inferred from the edges we have marked at each point.

As Ukkonen's, Weiner's algorithm is linear-time under the same assumptions of a transdichotomous RAM model with computer word size $\Omega(\log n)$; otherwise a factor of $O(\log |\Sigma|)$ multiplies the time complexity.

\begin{theorem}
There exists an algorithm to compute smallest suffixient sets that processes a $T\$$ right to left such that, for every $n$, after having processed the suffix $w[1\dd n]\$$ of $T\$$ it has determined a smallest suffixient set of $w\$$, and its size $\chi(w)$, using $O(n)$ space and $O(n)$ worst-case time.
\end{theorem}

\section{Conclusions and Open Questions}

We have contributed to the understanding of $\chi$ as a new 
measure of repetitiveness, better finding its place among more studied ones. Figure~\ref{fig:measures} shows the (now) known relations around $\chi$ (cf.~\cite{NavSurvey}). We also proved various additive and multiplicative lower and upper bounds on the sensitivity of $\chi$ to various string operations: appends/prepends, general edits, rotations, and reversals. As a direct consequence of our results, we derive new simple linear-time online algorithms to compute smallest suffixient sets, and thus $\chi$, as a modification of Ukkonen's and Weiner's suffix tree constructions \cite{Ukk95,Wei73}.

There are still many interesting open questions about $\chi$. As a follow-up to our conference paper, it has been recently shown that our bound $\chi \le 2r$ is tight \cite{tight2r}. It was also proved recently that $\chi$ is reachable \cite{SB26}, a major breakthrough. Nevertheless, it is still not known if $\chi = \Omega(b)$.

One consequence of Corollary~\ref{cor:chi_uncomparable_g} is that $\chi \not \in O(g\log^k(n/g))$ for any $k \ge 0$. 
It could be the case, though, that $\chi \in O(\delta\log n)$, because the separation of $\chi$ and $\delta$ on de Bruijn sequences is a $\Theta(\log n)$ factor.

Regarding edit operations, we proved a lower bound of $\Omega(\sqrt{n})$ to their additive sensitivity. A trivial upper bound is $O(n)$ because $\chi \le n$, but we conjecture that the sensitivity is indeed $\Theta(\sqrt{n})$. Proving or disproving this conjecture is an open problem (per our results, it suffices to prove the bound $O(\sqrt{n})$ for rotations, deletions, or substitutions). With respect to multiplicative sensitivity, experiments suggest that $\sre(w')/\sre(w)$ is $O(1)$ for all the string operations we considered.
We proved such a multiplicative constant for reversals. Proving the same for insertions would imply the existence of a constant for rotation and vice versa. It is also open whether $r = O(\chi\log n)$. If this were true  ---and provided that $\chi$ has $O(1)$ multiplicative sensitivity to string operations--- it would imply that $r$ has $O(\log n)$ multiplicative sensitivity to these operations, making the already known lower bounds on multiplicative sensitivity~\cite{AFI23,GILPST21,GILRSU25} tight. If the conjecture were false, then $\chi$ could be considerably smaller than $r$ in some string families. In the same line, observe that Lemma \ref{le:reverse_multiplicative}  implies that, either $r(w)/r(w^R)=\omega(\log n)$ and $r$ can be $\omega(\chi\log n)$, or $r(w)/r(w^R) = O(\log n)$ and this known bound is tight.

\subsection*{Acknowledgements}

We thank Davide Cenzato, Francisco Olivares, and Nicola Prezza, for useful discussions on suffixient sets, and for their code to compute smallest suffixient sets \url{https://github.com/regindex/suffixient}~\cite{cop:spire2024}, which was helpful to propose and discard hypotheses on the behavior of $\chi$, and to verify the implementation of our Ukkonen's extension that computes them online. We also thank Gregory Kucherov and Dominik K\"oppl for pointing out a bug in an early version of our Weiner's extension.

\subsection*{Disclosure of interests}The authors have no competing interests to declare that are relevant to the content of this article.

\subsection*{Funding}

Open access funding provided by Università degli Studi di Palermo within the CRUI-CARE Agreement.

H.F. was supported by JST BOOST, Japan Grant Number JPMJBS2406.

G.N. and C.U. were partially funded by Basal Funds FB0001 and \sloppy AFB240001, ANID, Chile; and Fondecyt grants 1230755 and 1260080, ANID, Chile.

G.R. was partially supported by the project \vir{ACoMPA – \sloppy Algorithmic and Combinatorial Methods for Pangenome Analysis} (CUP B73C24001050001) funded by
NextGeneration EU programme PNRR MUR M4 C2 Inv. 1.5 - Project ECS00000017 Tuscany Health Ecosystem (Spoke 6), CUP Master B63C22000680007, and by the INdAM - GNCS Project \sloppy CUP$\_$E53C25002010001. 

C.U. was supported by the Polish National Science Center, grant no. 2022/46/E/ST6/00463; \sloppy ANID-Subdirección de Capital Humano/Doctorado Nacional/2021-21210580, ANID, Chile; and NIC Chile Doctoral Scholarship, NIC, Chile.

\bibliographystyle{elsarticle-num} 
\bibliography{bibliography}

\end{document}